\documentclass{edm_template}

\usepackage{hyperref}
\usepackage{nameref}
\usepackage{url}
\usepackage{algorithm}
\usepackage[noend]{algpseudocode}
\usepackage{subfigure}
\usepackage{amsmath}
\usepackage{amsthm}
\usepackage{dcolumn}
\usepackage{graphicx}
\usepackage{algorithmicx}
\usepackage{algpseudocode}% http://ctan.org/pkg/algorithmicx
\usepackage{booktabs,siunitx}
\usepackage{amsfonts}

\usepackage{topcapt}
\usepackage{booktabs}
\usepackage{flushend}

\usepackage{graphicx}
\usepackage{caption}
\usepackage{setspace}
\usepackage{lipsum}
\usepackage{multicol}

\begin{document}

\title{The Guided Team-Partitioning Problem: Definition, Complexity, and Algorithm\vspace*{-10cm}}

\algdef{SE}[SUBALG]{Indent}{EndIndent}{}{\algorithmicend\ }%
\algtext*{Indent}
\algtext*{EndIndent}

%% Lemma, Theorem, Definiton
\newtheorem{problem}{Problem}
\newtheorem{definition}{Definition}
\newtheorem{observation}{Observation}
\newtheorem{lemma}{Lemma}
\newtheorem{corollary}{Corollary}

\newcommand{\squishlist}{\begin{list}{$\bullet$}
  { \setlength{\itemsep}{0pt}
     \setlength{\parsep}{3pt}
     \setlength{\topsep}{3pt}
     \setlength{\partopsep}{0pt}
     \setlength{\leftmargin}{1.5em}
     \setlength{\labelwidth}{1em}
     \setlength{\labelsep}{0.5em} } }
\newcommand{\squishend}{
  \end{list}  }
  
%% Problem Formulation
\newcommand{\inpSet}{{\ensuremath{\mathcal{R}}}}  %input dataset
\newcommand{\cluster}{{\ensuremath{\mathcal{C}}}} 
\newcommand{\inpElem}{{\ensuremath{\mathbf{r}}}} % elements of input dataset
\newcommand{\noRemove}{{\ensuremath{\mathbf{\ell}}}}  % number of elements we want to remove
\newcommand{\noTarget}{{\ensuremath{\mathbf{t}}}} % number of given target cluster reps

\newcommand{\std}{{\ensuremath{\mathbf{\sigma}}}}

\newcommand{\pointInCluster}{{\ensuremath{m}}}
\newcommand{\target}{{\ensuremath{\mathbf{t }}}}  %target vectors (or values)
\newcommand{\targetSet}{{\ensuremath{\mathbf{T }}}}  %set of target vectors (or values)
\newcommand{\calP}{{\ensuremath{\mathcal{P}}}} %obtained partisions after clustering
\newcommand{\subsetS}{{\ensuremath{\mathbf{S}}}}
\newcommand{\distanceFunction}{{\ensuremath{\mathbf{D}}}}
\newcommand{\mean}{{\ensuremath{\mathbf{mean}}}}

\newcommand{\noPoints}{{\ensuremath{n}}}
\newcommand{\noCluster}{{\ensuremath{k}}} % number of desired clusters to obtain
\newcommand{\dimension}{{\ensuremath{d}}}

%% datasets  and experiments
\newcommand{\DSskills}{{\tt Skills}}
\newcommand{\bia}{{\tt BIA}}
\newcommand{\guru}{{\tt Guru}}
\newcommand{\freelancer}{{\tt Freelancer}}
\newcommand{\synth}{{\tt Synthetic}}

\newcommand{\rep}{{\ensuremath{\mathbf{rep}}}}

\newcommand{\clusterCenters}{{\ensuremath{\mathbf{C}}}}
\newcommand{\clusterCenter}{{\ensuremath{\mathbf{c}}}}

%% algorithm and problem names
\newcommand{\LRemoval}{{\sc CIS}}  %l removal problem, selecting only one subset such that mean is close to target
\newcommand{\targetClustering}{{\sc CP}}  % target clustering problem, clustering data such that we have k clusters and obtained cluster centers are close to given target vectors
\newcommand{\targetLClustering}{{\sc Guided Team-Partitioning}}  %target clustering problem, clustering data such that we have k clusters and obtained cluster centers are close to given target vectors and we can remove up to l point
\newcommand{\greedy}{{\sc Greedy}} % greedy algorithm to solve LRemoval problem
\newcommand{\finalAlgo}{{\sc GuidedSplit}}
\newcommand{\btfCvx}{{\sc BTF\textunderscore CVX}}
\newcommand{\btfGreedy}{{\sc BTF\textunderscore Greedy}}
\newcommand{\assignment}{{\sc Assignment problem}}
\newcommand{\randomPartitioning}{{\sc Random}}
\newcommand{\kmeans} {{\sc k\_means}}
\newcommand{\kmeanst} {{\sc k\_targets}}
\newcommand{\kmeansmin} {{\sc k\_means-{}-}}
\newcommand{\maxBenefit} {{\sc Max Benefit}}
\newcommand{\knn} {{\sc kNN}}
\newcommand{\cvx} {{\sc ConvexOpt}}

\newcommand{\calX}{{\ensuremath{\mathcal{X }}}}
\newcommand{\calY}{{\ensuremath{\mathcal{Y }}}}
\newcommand{\calU}{{\ensuremath{\mathcal{U }}}}
\newcommand{\calS}{{\ensuremath{\mathcal{S }}}}
\newcommand{\calH}{{\ensuremath{\mathcal{H }}}}
\newcommand{\calC}{{\ensuremath{\mathcal{C }}}}

\newcommand{\calA}{{\ensuremath{\mathcal{A}}}}
\newcommand{\centers}{{\ensuremath{\mathcal{C}}}}

\newcommand{\hatC}{{\ensuremath{\widehat{C}}}}

\newcommand{\ttu}{{\mathrm{\tt{u}}}}
\newcommand{\hatsigma}{{\ensuremath{\widehat{\sigma}}}}
\newcommand{\hatd}{{\ensuremath{\widehat{d}}}}
\newcommand{\bM}{{\ensuremath{\mathbf{M}}}}
\newcommand{\br}{{\ensuremath{\mathbf{r}}}}
\newcommand{\bB}{{\ensuremath{\mathbf{B}}}}
\newcommand{\bb}{{\ensuremath{\mathbf{b}}}}
\newcommand{\bm}{{\ensuremath{\mathbf{m}}}}
\newcommand{\bS}{{\ensuremath{\mathbf{S}}}}
\newcommand{\bT}{{\ensuremath{\mathbf{T}}}}
\newcommand{\bR}{{\ensuremath{\mathbf{R}}}}
\newcommand{\benefit}{{\ensuremath{\mathbf{b}}}}
\newcommand{\singschedule}{{\ensuremath{\mathbf{R}}}}
\newcommand{\distschedule}{{\ensuremath{\mathbf{D}}}}

\newcommand{\numt}{{\ensuremath{\#\mathbf{t}}}}
\newcommand{\numtp}{{\ensuremath{\#\mathbf{t^P}}}}
\newcommand{\numtpd}{{\ensuremath{\#\mathbf{t^P_d}}}}
\newcommand{\numtpi}{{\ensuremath{\#\mathbf{t^P_i}}}}
\newcommand{\numts}{{\ensuremath{\#\mathbf{t^s}}}}
\newcommand{\numtsi}{{\ensuremath{\#\mathbf{t^s_i}}}}
\newcommand{\numtsd}{{\ensuremath{\#\mathbf{t^s_d}}}}
\newcommand{\numtps}{{\ensuremath{\#\mathbf{t^{P,s}}}}}
\newcommand{\numtpsd}{{\ensuremath{\#\mathbf{t^{P,s}_d}}}}
\newcommand{\numtpsi}{{\ensuremath{\#\mathbf{t^{P,s_i}}}}}

\newcommand{\pos}{{\tt pos}}
\newcommand{\I}{{\tt I}}

\newcommand{\req}{{\tt req}}

\newcommand{\dU}{{\ensuremath{\mathbf{\Delta U}}}}

\newcommand*{\argminl}{\argmin\limits}
\newcommand*{\argmaxl}{\argmax\limits}

\newcommand{\singlet}{{\tt Single\_t}}
\newcommand{\group}{{\tt Group}}
\newcommand{\groupt}{{\tt Group\_t}}
\newcommand{\evalp}{{\tt Evaluate\_P}}
\newcommand{\evalpr}{{\tt Evaluate\_Pr}}
\newcommand{\evalpl}{{\tt Evaluate\_PL}}
\newcommand{\greedyt}{{\tt Greedy\_t}}
\newcommand{\partdp}{{\tt Partition\_dp}}
\newcommand{\partdpl}{{\tt Partition\_dpl}}
\newcommand{\parte}{{\tt Partition\_e}}
\newcommand{\select}{{\tt Select}}
\newcommand{\schedule}{{\tt Schedule}}
\newcommand{\kmins}{{\tt K-Mins}}
\newcommand{\kminsm}{{\tt K-Mins-{}-}}
\newcommand{\kmeanspp}{{\tt K-Means+{}+}}

%---------------------------------------

\newcommand{\asdata}{{\tt AS}}
\newcommand{\wikivote}{{\tt WikiVote}}
\newcommand{\dblp}{{\tt DBLP}}

%% editing macros
\newcommand{\spara}[1]{\smallskip\noindent{\bf{#1}}}
\newcommand{\mpara}[1]{\medskip\noindent{\bf{#1}}}
\newcommand{\bpara}[1]{\bigskip\noindent{\bf{#1}}}

\title{The Guided TeamPartitioning Problem: Definition, Complexity, and Algorithm
\titlenote{(Does NOT produce the permission block, copyright information nor page numbering). For use with ACM\_PROC\_ARTICLE-SP.CLS. Supported by ACM.}}

\numberofauthors{3} %  in this sample file, there are a *total*
\author{
% 1st. author
\alignauthor
Sanaz Bahargam\\
       \affaddr{Boston university}\\
       \email{bahargam@bu.edu}
% 2nd. author
\alignauthor
Theodoros Lappas\\
       \affaddr{Stevens Institute of Technology}\\      
       \email{tlappas@stevens.edu}
\alignauthor Evimaria Terzi\\
       \affaddr{Boston University}\\
       \email{evimaria@cs.bu.edu}
}
% There's nothing stopping you putting the seventh, eighth, etc.
% author on the opening page (as the 'third row') but we ask,
% for aesthetic reasons that you place these 'additional authors'
% in the \additional authors block, viz.
\additionalauthors{Additional authors: John Smith (The Th{\o}rv{\"a}ld Group,
email: {\texttt{jsmith@affiliation.org}}) and Julius P.~Kumquat
(The Kumquat Consortium, email: {\texttt{jpkumquat@consortium.net}}).}
\date{30 July 1999}
% Just remember to make sure that the TOTAL number of authors
% is the number that will appear on the first page PLUS the
% number that will appear in the \additionalauthors section.

\maketitle
\vspace*{-3.5em}
\begin{abstract}
A long line of literature has focused on the problem of selecting a team of individuals from a large pool of
candidates, such that certain constraints are respected, and
a given objective function is maximized. Even
though extant research has successfully considered diverse
families of objective functions and constraints, one of the
most common limitations is the focus on the single-team
paradigm. Despite its well-documented applications in multiple
domains, this paradigm is not appropriate when the
team-builder needs to partition the entire population into
multiple teams. Team-partitioning tasks are very common in
an educational setting, in which the teacher has to partition
the students in her class into teams for collaborative
projects. The task also emerges in the context of organizations,
when managers need to partition the
workforce into teams with specific properties to tackle relevant projects. In this work, we extend
the team formation literature by introducing the Guided
Team-Partitioning (GTP) problem, which asks for the partitioning of a population
into teams such that the centroid of each team is as close
as possible to a given target vector.  As we describe in detail in our work, this formulation allows the team-builder to control the composition of the produced 
teams and has natural applications in practical settings. Algorithms for the GTP need to 
simultaneously consider the composition
of multiple non-overlapping teams that compete for
the same population of candidates. This makes the problem
considerably more challenging than formulations that focus
on the optimization of a single team. In fact, we prove that
GTP is NP-hard to solve and even to approximate. The
complexity of the problem motivates us to consider efficient
algorithmic heuristics, which we evaluate via experiments
on both real and synthetic datasets.
\end{abstract}

%% A category with the (minimum) three required fields
%%A category including the fourth, optional field follows...
%\category{D.2.8}{Software Engineering}{Metrics}[complexity measures, performance measures]

 % NOT required for Proceedings

The input of the  general team formation problem consists of a pool of candidates, a set of constraints, and an objective function. The goal is then to strategically select 
a group of individuals from the pool, such that the group respects all the constraints and also optimizes the objective function. A long line of literature has addressed different versions of the problem
that ask for the optimization of functions such as the quality of intra-team communication~\cite{Lappas:2009}.
% the ...  \hl{ADD AS MANY AS YOU CAN}. 
Previous work has also considered a diverse array of constraints, such as those on the team's size~\cite{walter2016minimizing}, the team-builder's recruitment budget~\cite{golshan2014profit}, the coverage of a particular set of skills~\cite{tang2016profit}, and the workload allocated to each team member~\cite{hendriks1999human}.

In this paper, we extend the team formation literature by introducing an alternative formulation of the problem that asks for the partitioning of the given pool of candidates into \emph{multiple} teams, while
allowing the team builder to control the properties of each team. We refer to this paradigm as \emph{Guided Team-Partitioning} (GTP). For instance, consider a teacher who is trying to partition the students in her class into groups, such that the distribution of talent and experience across the groups is balanced~\cite{shaw2004fair}. We illustrate an example of this scenario in Figure~\ref{fig:intro}. 
\begin{figure}[htb]
\centering
\includegraphics[scale=0.15]{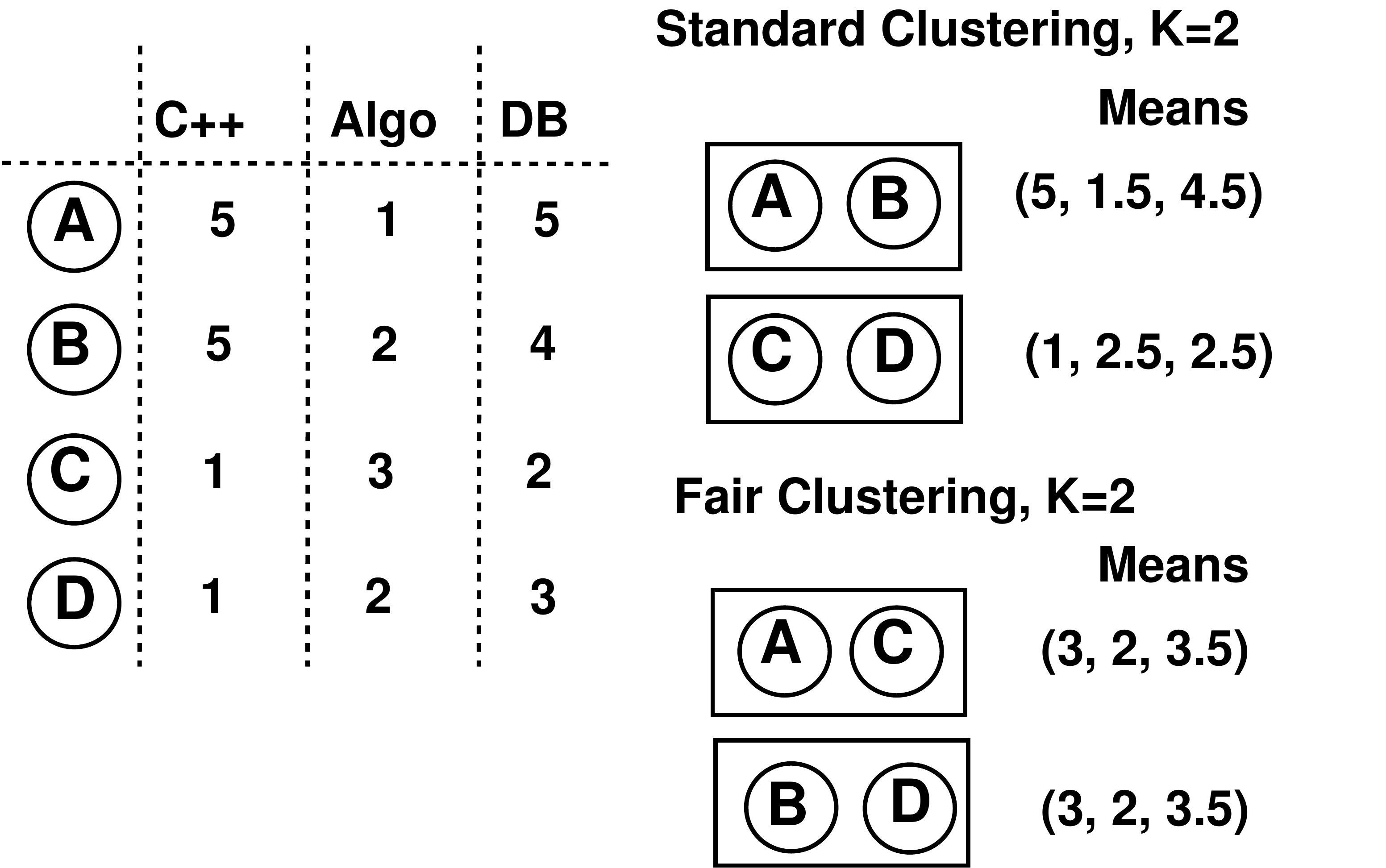}
  \vspace{-9pt}
    \caption{Standard Vs. Fair Clustering}
      \vspace{-9pt}
    \label{fig:intro}
\end{figure}
In this example, the input consists of a population of $4$ students $\{A, B, C, D\}$. Each student has a proficiency level for three skills: C++, Algorithms, and Databases. Proficiency is measured on a 1-5 scale, with higher values indicating higher proficiency. The goal is to partition the student population into balanced teams, such as each team has an average proficiency level of $3$ for all three skills. Conceptually, our goal
is to avoid teams with an unfair advantage or disadvantage. We refer to $(3, 3, 3)$ as the \emph{target vector}, which is used to guide the partitioning task.
One approach would be to use a similarity-based clustering algorithm, 
such as K-means. The best solution would then be to create two clusters $\{A, B\}$ and $\{C, D\}$. The students in the first team would then be highly similar to each other, being experts in C++ and Databases but novices with respect to Algorithms. On the other hand, both $C$ and $D$ have very little knowledge of C++ and are similarly mediocre with respect to the other two skills. Therefore, it is clear that this partitioning 
fails to approximate the target vector and does  not deliver a balanced distribution of talent. If we consider all possible team assignments, it is clear that the best possible option is to create two teams $\{A, C\}$ and $\{B, D\}$ with both having (3, 2, 3.5) as their centroids.  The two identical centroids clearly demonstrate the fairness of this partitioning: both teams have an average proficiency of $3, 2$, and $3.5$ for C++, Algorithms, and Databases, respectively. In addition, the $(3, 2, 3.5)$ vector is the closest possible to the ideal $(3,3,3)$ target vector.

This first example captures the scenario in which the centroids of all the teams have to be as close as possible to the same target vector.
What if the team builder wants to go beyond this simple balanced partitioning and actually ask for \emph{different} centroid values for the produced teams? Consider the example shown in Figure~\ref{fig:intro2}.

\begin{figure}[htb]
\centering
\includegraphics[scale=0.12]{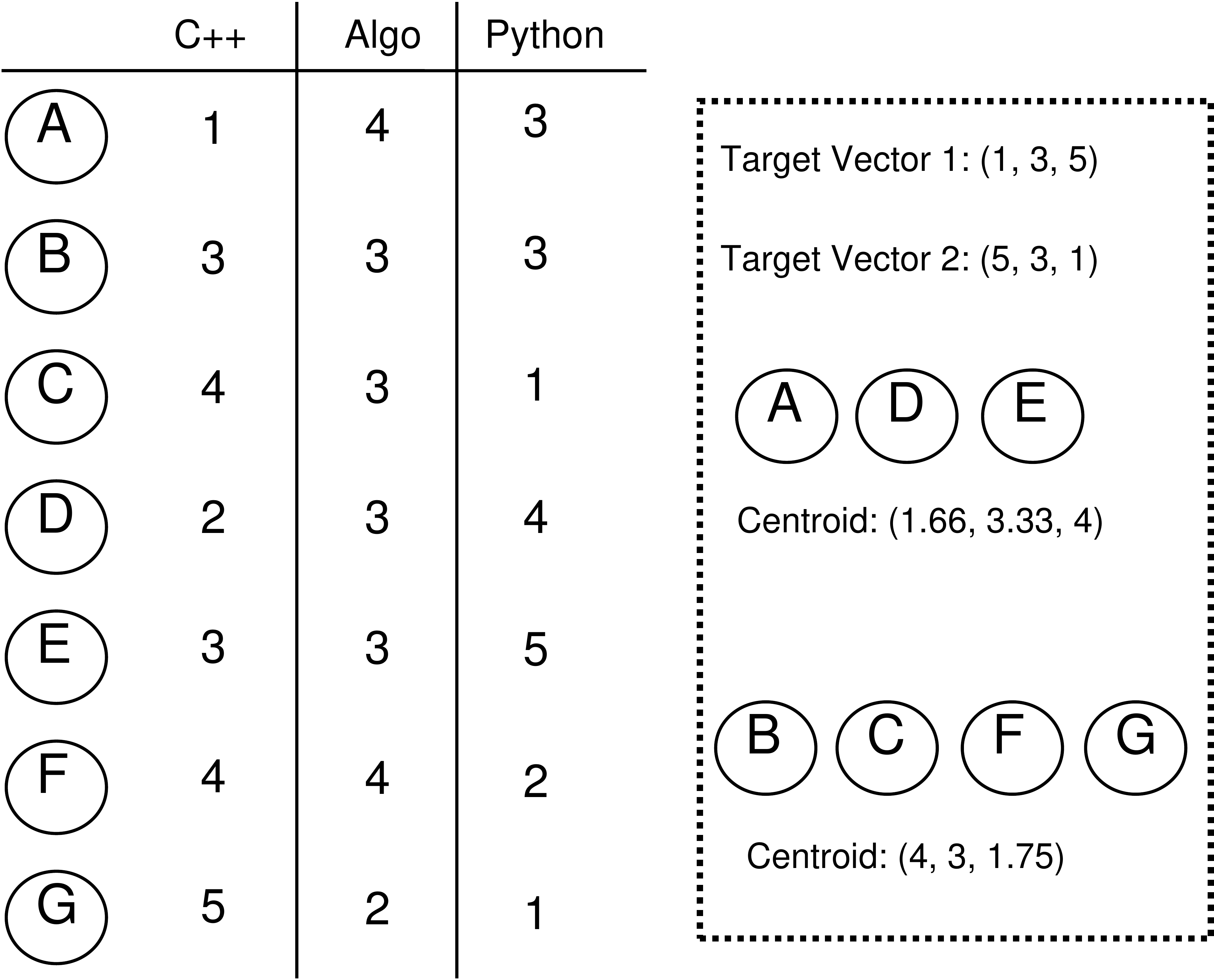}
  \vspace{-9pt}
    \caption{Guided partitioning of a population into teams.}
      \vspace{-9pt}
    \label{fig:intro2}
\end{figure}

The figure shows a population of 7 students. For each student, we know their proficiency level with respect to C++, Algorithms, and Python, measured on a 1-5 scale. The teacher wants to compare
C++ and Python in the context of collaborative projects. In order to achieve this, she needs to partition her 7 students into two teams, such that one team has a high concentration in C++ and a low
concentration in Python, while the other team has the exact opposite properties. In addition, both teams need to have a medium concentration in Algorithms, in order to control for the effect of this 
skill~\footnote{In a more general example, the teacher would ask for a medium level for all available properties (e.g. demographics, proficiency in other relevant skills) except Python and C++}. 
In the context of the (C++, Algorithms, Python) feature space, the teacher thus specifies the target vectors $(1, 3, 5)$ and $(5, 3, 1)$. An exhaustive evaluation of all possible solutions would
reveal that the best partitioning would be $(A, D, E)$ and $(B, C, F, G)$. As we show in the picture, the centroids of these two teams are as close to the specified target vectors as possible. 
%We formalize our framework according to the more general formulation of this second example, which allows the team-builder to specify a different target vector for each team. 
This second example demonstrates the main problem that we focus on in our work: 

\spara{\tt Main Problem:}
\emph{Given a population of individuals, we want to partition the population into teams such that the centroid of each team is as close as possible to a specific target vector.}
\ \\\\
In an educational setting, an instructor could use this type of partitioning to create customized training environments by placing students in teams with specific strengths and weaknesses~\cite{bailey2005teaching,razzouk2003learning}. In general, this type of team partitioning is essential for a team builder who is trying to form and study groups with specific compositions for marketing or experimental purposes~\cite{baugh1997effects}. We observe that, even for the toy examples that we presented above, finding the best possible partitioning is a non-trivial task. In addition, the task obviously becomes even more challenging in realistic scenarios with large numbers of candidates and features. In fact, as we show in our work, the problem is NP-hard to solve and even to approximate.

While there exist team formation problems that try to form good teams to cover some skills while minimizing the cost, to the best of our knowledge, our work is the first to tackle this customized team formation problem. Nonetheless, our work has ties to extant research on team formation and other problems, which we review in
Section~\ref{sec:related}. In Section~\ref{sec:problem}, we explore alternative formal definitions and study their hardness. We then present
an efficient algorithmic framework in Section~\ref{sec:algo}. In Section~\ref{sec:experiments}, we evaluate the efficacy of our framework via an experimental evaluation that includes both
real and synthetic data, as well as competitive baselines. Finally, we conclude our work in Section~\ref{sec:conclusion}.

\section{Related Work}\label{sec:related}
Although to best of our knowledge, we are the first to introduce the \targetLClustering\ Problem, the nature of our problem is related to semi-supervised clustering  and team formation problem. We review only some of these works here:

\spara{Semi-supervised Clustering}
Existing methods for semi-supervised clustering fall into
two general approaches constraint-based and
metric-based. In constraint-based approaches, the clustering
algorithm itself is modified to lead the algorithm
towards a more appropriate data partitioning. In order to achieve this, pairwise constraints or user-provided labels
are used in the algorithm. This is done by initializing and constraining clustering
based on labeled examples ~\cite{Basu02semi-supervisedclustering}, modifying objective function to include constraints~\cite{demiriz1999semi}, or enforcing constraints during the clustering process~\cite{wagstaff2001constrained}. In semi-supervised clustering by seeding~\cite{Basu02semi-supervisedclustering}, besides input dataset \inpSet\ and the number of clusters $k$, given is a subset of the dataset, $S$, in which the cluster they should belong is specified. For all $k$ clusters, there is at least one point in $S$. 
One of the earliest problems which is close to constraint-based clustering is  facility location problem~\cite{Shmoys:1997:AAF}  and it is studied mainly in operational research science. It tries to locate $k$ facilities to serve $n$ customers such that the travelling distance from the customers to their facility is minimized. However, the only type of constraints they studied is constraints on the capacity of the facility, i.e., each facility can only serve a limited number of customers. 
In the work by Bradley et al. \cite{Bradley00constrainedk-means},  $k$ constraints are added to the underlying clustering optimization problem requiring that each cluster has at least a minimum number of points in it. 
Tung et al.~\cite{Tung00constraint-basedclustering} introduces a framework for constraint-based clustering. Their taxonomy of constraints includes constraints on individual objects (e.g. cluster only luxury mansions of value over one million dollars), parameter constraints (e.g. the number of clusters) and constraints on individual clusters that can be described in terms of bounds on aggregate functions (min, avg, etc.).
On the contrary, in metric-based approaches, an existent clustering algorithm
that uses a distance metric is used; however, the metric
is first trained to satisfy the labels or constraints in the given
supervised data. Various distance measures have been used
for metric-based approaches including but not limited Euclidean
distance trained by a shortest-path algorithm \cite{klein2002instance}, Jensen-Shannon divergence trained using gradient descent~\cite{cohn2003semi}, string-edit distance learned using Expectation
Maximization~\cite{Bilenko:2003},  or Mahalanobis distances trained using convex optimization~\cite{bar2003learning}.

\spara{Team Formation} Team formation has been studied in operations research community e.g. ~\cite{Baykasoglu:2007}, which defines the problem as finding the optimal match between  people and demanded functional requirements. It is often solved using techniques such as simulated annealing, branch-and-cut or genetic algorithms e.g. ~\cite{Baykasoglu:2007}. 
Lately it has also been studied in computer science. 
%~\cite{Lappas:2009,bahargam2015personalized,Anagnostopoulos:2010,Kargar:2011,Eftekhar:2015,agrawal2014forming,BasuRoy:2015,wang2016ustf}
%Li2012CAG}. 
A majority of these work focus on team formation to complete a task and minimize the communication cost among team members~\cite{BAHARGAM2019441}. The focus of these studies is to find only one team to  perform a given task. Lappas et al. ~\cite{Lappas:2009} introduces the problem of team formation in the context of social networks. Given a pool of experts and a set of skills that needed to be covered, the goal there is to select a team of experts that can collectively cover all the required skills while ensuring efficient intra- team communication. Their work  imposes the strong assumption that a single person can fulfill a skill requirement of a task. Whereas a general framework can impose the constraint such that the at least $n$ experts should fulfill the requirements \cite{pragarauskas2012multi}.
Bhowmik et al.  \cite{bhowmik2014submodularity} developed algorithms using submodularity to find teams of experts by relaxing the skill cover requirement such that some skills must be necessarily covered by experts while other skills only improve the team quality. 
The work by Rangapuram et al. \cite{Rangapuram:2013} also studies the problem of finding a team of experts based on densest subgraph  that is both competent in performing the given task and compatible in working together. 
%\cite{awal2014team} tries to identify a set of experts that can collaborate effectively to complete a given task.  
The recent work by Anagnostopoulos et al. \cite{anagnostopoulos2012online} considers a time-series of arriving tasks whereby users are chosen to finish the arriving tasks without overwhelming any team or any expert, and the team has small communication overhead. A similar study by Kargar et al. \cite{kargar2013finding} considers the problem of finding an affordable and collaborative team from an expert network that minimizes two objectives at the same time: the communication cost among team members and also the personnel cost.
Recently,   complex task crowdsourcing by team formation has been studied, where the requester wishes to hire a group of workers to work together as a team \cite{wang2016truthful}. 
In educational settings clustering students into different teams is studied such that students can maximally benefit from peer interaction \cite{bahargam2015personalized, DBLP:journals/corr/BahargamEBT17}.   In addition, the work by Agrawal et al. \cite{Agrawal:2014} considers partitioning students in which each student has only one ability level for all the activities, and each team has a set of leaders and followers in which leaders are helping followers to complete a task. 
The goal is to maximize the gain of students where the gain is defined as  the number of students who can do better by interacting with the higher ability students. 

%There is also  literature discussing the structure of teams and the influence of team structure on team performance ~\cite{}

%Although all these works focus on identifying good teams, they are different from the work we present here,  as most of them only focus on finding only one team \cite{}, focusing on binary skills \cite{}, or only considering one ability level \cite{} or they consider different objective function \cite{} .as we introduce a new optimization function; forming customized teams by placing individuals in teams with specific strengths and weaknesses. 
Although all these works focus on identifying good teams, they are different from the work we present here,  as most of them focus on finding only one team, or only allow for binary skills, or taking into account only one ability level,  or they study very  different objective function. In this paper, we introduce a new different  optimization function; forming customized teams by placing individuals in teams with specific strengths and weaknesses. 

\section{Problem Definition}\label{sec:problem}
In this section, we begin by describing notational conventions that we will use throughout the paper; then we present the formal statement of the
problems that we study. We start from a simple version of our problem with only one team (partition). We show that even the simple version is NP-hard to solve and approximate. Then we move to partitioning problem  with desired centroids.  In Subsection~\ref{subsection:final_prob}, we describe our problem, \targetLClustering. In the same subsection, we show that  our problem is NP-hard to solve and approximate.

\subsection{Preliminaries}
We consider a pool {\inpSet} of \noPoints\ individual experts. Each expert $\inpElem \in \inpSet$ is associated with a $\dimension$-dimensional feature (skill) vector $\inpElem_i$, such that $\inpElem_i(f)$ returns the value of skill $f$ for expert $i$. We also consider given a set of target vectors \targetSet\ of \noCluster\ target points \target. The goal is to partition the pool of experts into \noCluster\ teams such that the cost of this partitioning is minimized. The cost for each team $\cluster_i$,  is the distance between the centroid of that team (\mean($\cluster_i$))  to its target vector  $\target_i$.
\begin{equation}\label{eq:cost}
Cost = \sum\limits_{i=1}^\noCluster \distanceFunction( \mean(\cluster_{i}), \target_{i} )
\end{equation}

We quantify the closeness between two vectors $\tau$ and $t$  of dimension $d$, using the $L^{2}_{2}$ norm of their difference. We denote this by
$$ \distanceFunction(\tau, t ) := L^{2}_{2} (\tau - t ) = \sum\limits_{i=1}^d  (\tau(i) - τt(i))^{2}$$

We define the \mean\ of a set \inpSet, consisting of $n$ vectors $\inpElem_1, \inpElem_2, \ldots, \inpElem_n$  as  $$\mean(\inpSet) = \frac{\sum_{i=1}^{n} \inpElem_i}{n}$$

\iffalse
\note[SB]{ Can it also work for Manhattan or Jaccard distance or any metrics? Cosine distance?} 
\fi 

%\subsection{}
%Note that the problem we are going to introduce are not clustering problems. They are ..

\subsection{The \LRemoval\  Problem}
In Characteristic-Item Selection (\LRemoval) problem, the goal is to find a subset of a universe set, such that
the mean of the subset is as close as possible to a designated point. This designated point represents the whole set. This intuition is captured in a formal definition as follows. 

\begin{problem}[\LRemoval]\label{prob:groupschedule}
Given a set $\inpSet= \lbrace \inpElem_{1}, \inpElem_{2}, \ldots, \inpElem_{\noPoints} \rbrace$, a number $\noRemove$, and a designated vector $\target$, find \noRemove\ points ($\inpSet_\noRemove$) to remove from \inpSet, such that  \mean\ of the remaining points $\inpSet \setminus \inpSet_{\noRemove}$ is close to target \target. More formally,
%$$||mean(\inpSet \setminus \inpSet_{\noRemove}), \target ||_{2}^{2}$$
$$\distanceFunction(\mean(\inpSet \setminus \inpSet_{\noRemove}), \target )$$
is minimized.
\end{problem}

\begin{lemma}
The \LRemoval\ problem is NP-hard to solve and approximate.
\end{lemma}

%We refer the reader to the supplementary material for the proof. 
Our analysis below demonstrates that the \LRemoval\ problem is not only NP-hard, but it is also NP-hard to approximate. In order to see this, let's consider the decision version of the problem which is defined as follows: Given a set $\inpSet= \lbrace \inpElem_{1}, \inpElem_{2}, \ldots, \inpElem_{\noPoints} \rbrace$, a number \noRemove, and a designated vector $\target$, does there exist a subset  $\inpSet_\noRemove \subseteq \inpSet$ with $|\inpSet_\noRemove| = \noRemove$ and $\distanceFunction(\mean (\inpSet \setminus \inpSet_\noRemove ), \target) \leq q$? We call this decision version of the problem Decision-\LRemoval. The following lemma shows that this problem is NP-complete.

\begin{lemma}
\label{l_removal_np}
The Decision-\LRemoval\ problem is NP-complete.
\end{lemma}

\begin{proof} \label{proof:1}
Without the loss of the generality, we show the NP-completeness of the \LRemoval \ problem
for the case when the elements in $\inpSet$ are one-dimensional.
We will reduce an instance of the Subset Sum problem to \LRemoval. An instance of Subset Sum problem states that given a universe of $n$ numbers $U=\{x_1,x_2, \ldots, x_n\}$, does there exist a subset of size $j$ whose sum is $J$? We translate an instance of Subset Sum to \LRemoval \ by setting $\inpSet$ to $U$, $\noRemove = \noPoints - j$ and $\target$ to $\frac{J}{\noPoints - \noRemove} $ and $q = 0$.
We can decide the Subset Sum problem
if and only of we can decide for \LRemoval\ problem.\end{proof}

\begin{corollary}
The \LRemoval\ problem is NP-hard and NP-hard to approximate. That is, it is NP-hard to find a polynomial-time approximation algorithm for the \LRemoval\ problem.
\end{corollary}

\begin{proof}
We will prove the hardness of approximation of the \LRemoval\ problem by contradiction. Assume that there exists an $\alpha$-approximation algorithm for the \LRemoval\ problem. Then if $S^*$ is the optimal solution to the problem and $S^A$ is the solution output by this approximation algorithm, it will hold
that $  \distanceFunction( \mean(S^{A} ), \target ) \leq \alpha  \distanceFunction( \mean(S^{*} ), \target ))$. 
 If such an approximation algorithm exists, then this algorithm can be used to decide the decision instances of the \LRemoval\  problem for which $q = 0$. However, this contradicts the proof of Lemma~\ref{l_removal_np}, which indicates that these problems are also NP-hard. Thus, such an approximation algorithm does not exist.
\end{proof}

As a use case of this problem, we can name the Characteristic-Review Selection problem (CRS) in~\cite{lappas2012selecting} which aims to find a small set of reviews from a corpus of reviews for an item, such that the subset 
presents the whole corpus the best. This is highly useful in on-line shopping and review websites. Consider Tripadvisor;  users can rate the characteristics of a hotel such as value, sleep quality, etc. When a user searches for a  hotel, she would like to see only an informative subset of reviews which represent all reviews the best and the mean of each feature in the selected subset is as close as possible to the mean of features in the whole set of reviews.

\subsection{The \targetClustering \ Problem} 
Given a set of $n$ workers with different skills and a set of $k$ tasks that need to be done collaboratively
with given average required skill levels for each task, the goal is to form teams of workers that their mean skill level matches each task the best.
This team formation problem can be considered as one of the applications of the $\targetClustering$ (Characteristic Partitioning) problem which is formally defined in the following problem

\begin{problem}[\targetClustering]\label{prob:targetCluster}
Given a set $\inpSet = \lbrace \inpElem_{1}, \inpElem_{2}, \ldots, \inpElem_{\noPoints} \rbrace$ and target vectors $\targetSet = \lbrace \target_{1}, \target_{2}, \ldots, \target_{\noCluster} \rbrace$,  partition $\inpSet$ into $\noCluster$ partitions $\cluster_1, \cluster_2, \ldots, \cluster_\noCluster$ such that
%$$\sum\limits_{i=1}^\noCluster \sum\limits_{\inpElem \in C_{i}}  || \inpElem - \target_{i}||_{2}^{2}$$
%$$\sum\limits_{i=1}^\noCluster || \mu(C_{i}) - \target_{i}||_{2}^{2}$$
$$\sum\limits_{i=1}^\noCluster \distanceFunction( \mean(\cluster_{i}), \target_{i} )$$
is minimized. 
\end{problem}
%
%Question:  Is the aforementioned cost function equal to
%$$\sum\limits_{i=1}^\noCluster |C_{i}| * || \inpElem - \target_{i}||_{2}^{2}$$
%
%Another interesting cost function is:
%$$\sum\limits_{i=1}^\noCluster || \mu(C_{i} - \target_{i}||_{2}^{2}$$

\begin{corollary}
The \targetClustering\ problem is NP-hard to solve and NP-hard to approximate.
\end{corollary}
%For the proof of complexity of \targetClustering\ problem, we refer the reader to supplementary material.
\begin{proof}
We will reduce an instance of Subset Sum problem to \targetClustering\ problem, by setting \noCluster=2 and $\target_1=\frac{J}{j}$. Thus if Cluster 1 has exactly $j$ points which the mean of points is $\frac{J}{j}$, the sum of elements in cluster 1 would be equal to $J$.  Without loss of generality let's assume for all $\inpElem_i \in \inpSet$, $0 \leq \inpElem_i \leq v$. Now we can add an extra point $\inpElem_{n+1}$ such that $ v \ll \inpElem_{n+1}$. Now we can set $\target_2=\frac{\inpElem_{n+1} +  Sum(\inpSet) - J }{\noPoints - j + 1}$. In order to minimize the cost of \targetClustering\ problem, there should be $\noPoints - j$ elements of \inpSet\ and $\inpElem_{n+1}$  in cluster 2 and hence $j$ elements in Cluster 1. 
\end{proof}

\iffalse
\begin{lemma}
\label{target_clustering_np}
The \targetClustering\ problem is NP-hard.
\end{lemma}
\begin{proof}
We will reduce an instance of \LRemoval\ problem to \targetClustering\ problem. We translate an instance of \LRemoval\ to \targetClustering\  by setting $\noCluster = 1$. Now we can find the optimal solution for the \LRemoval\  problem if and only of we can find the best solution for \targetClustering\ problem.
\end{proof}
\fi

\subsection{The \targetLClustering\ Problem} \label{subsection:final_prob}
For many applications, it is not necessary to assign \emph{all} the points in the dataset to teams. 
For instance, when separating a workforce into teams so that each team has a specific level of expertise in each skill, it is acceptable to exclude some of the workers.
In general, having the option to exclude up to a fixed number of points adds flexibility and can only make the problem easier to solve. We capture this intuition 
in a formal problem definition as follows.

\begin{problem}[\targetLClustering] \label{section:final_prob}
Given a set $\inpSet = \lbrace \inpElem_{1}, \inpElem_{2}, \ldots, \inpElem_{\noPoints} \rbrace$, and target vectors $\targetSet = \lbrace \target_{1}, \target_{2}, \ldots, \target_{\noCluster} \rbrace$, and a budget $\noRemove$, find $\noRemove$ points ($\inpSet_\noRemove$) to remove from the dataset and partition the rest of the points $\inpSet \setminus \inpSet_{\noRemove}$ into $\noCluster$ partitions $\cluster_1, \cluster_2,  \ldots, \cluster_\noCluster$ such that
%$$\sum\limits_{i=1}^\noCluster || \mu(C_{i}) - \target_{i}||_{2}^{2}$$
Cost = $$\sum\limits_{i=1}^\noCluster \distanceFunction( \mean(\cluster_{i}), \target_{i} )$$
is minimized.
\end{problem}

\begin{corollary}
The \targetLClustering\ problem is NP-hard to solve and NP-hard to approximate.
\end{corollary}
This problem is clearly NP-hard, as it contains the \targetClustering\ problem as a special case (for $\noRemove=0$).
%
\iffalse
\begin{lemma}
\label{target_clustering_removal_np}
The \targetLClustering\ problem is NP-hard.
\end{lemma}

\begin{proof}
It is easy to see \targetLClustering is NP-hard. We will reduce an instance of \targetClustering problem to \targetLClustering problem. We translate an instance of \targetClustering to \targetLClustering  by setting $\noRemove = 0$. Now it is easy to see that we can find the optimal solution for the \targetClustering  problem if and only of we can find the best solution for \targetLClustering problem.
\end{proof}
\fi
%

So far we have discussed the \targetLClustering\ problem.  For a collection of workers \inpSet\ in which each worker has a proficiency level for each skill, the most natural translation of the target vectors  is the mean of the proficiencies or the required skills of a given project. For instance, consider online labor markets such as Freelancer
(\url{www.freelancer.com}), Guru (\url{www.guru.com}), and  oDesk (\url{www.odesk.com}) where employers hire freelancers with specific skills to work on different types of projects. The required skills listed for each project in these platforms can be used as the target vectors. In such a setting, each team should poses a specific share of expertise across all skills that makes the team capable of finishing a particular task or project.

\textbf{Discussion:} 
As a natural result of our objective function, as the number of partitions grows, the cost increase as well. 
In some applications, target vectors are not given, and in such a setting  the goal is to form fair teams with the same level of proficiencies. An example is partitioning the students into teams with the same abilities. If $\noCluster$ is not predetermined either, finding the right $\noCluster$ might not be easy due to the fact that as $\noCluster$ increases, the cost may increase as well. In this situation, we are not able to increase the number of partitions $\noCluster$, and stop when the cost is not improved. To overcome this, one can redefine the cost as the following:
$$\sum\limits_{i=1}^\noCluster \distanceFunction( \mean(\cluster_{i}), \target_{i} ) *|C_{i}|$$

%As a natural result of our objective function, as the number of clusters grows, the cost increase as well. If  more clusters are desired and each time the cluster center is getting far from the desired target; the overall cost should also increase. However if we want to normalize the cost (such that it doesn't increase as the number of clusters grows), we can define the cost as $$\sum\limits_{i=1}^\noCluster \distanceFunction( \mean(C_{i}) , \target_{i} ) *|C_{i}|$$ this means that for each cluster we are paying a cost of being far from the desired target multiplied by number of people in that cluster. This is only useful when we want to partition students into teams with equal abilities. In such a case, we can even try different number of  clusters, and choose the $\noCluster$ which minimizes the cost. 

\section{Algorithm}\label{sec:algo}
\newcommand{\mbr} {{\sc MBR}}
In this section, we describe \finalAlgo; it finds an efficient solution for \targetLClustering. We start by presenting algorithms to solve \LRemoval\ and \targetClustering\ problem and then we use these algorithms to solve \targetLClustering\ problem.
%We first give some background on kmeans, random partitioning, assignment problem and Hungarian algorithm.
%Then we proceed  by presenting algorithm to solve \LRemoval and \targetClustering problem. Then we used those algorithms to solve  the \targetLClustering problem.

\subsection{Solving The \LRemoval\ Problem}
Although \LRemoval\  problem is NP-hard to solve and approximate,  we propose two heuristic algorithms that work well in practice. 

\spara{The \greedy\ Algorithm:}
The \greedy\ algorithm  is an iterative algorithm to solve \LRemoval\ problem. At each iteration, the algorithm selects a point to remove from data, such that it decrease the distance of  \target\ from the mean of remaining points. The pseudocode of the \greedy\ algorithm is shown in Algorithm \ref{algo:greedy}.

\vspace{-7pt} 
\begin{algorithm}[ht!]
\scriptsize 
\begin{algorithmic}[1]
\Statex {\bf Input:}  Input set \inpSet, target vector \target\ and number of points to be removed \noRemove
\Statex {\bf Output:}  A subset $\inpSet_{\noRemove}$ such that $|\inpSet_{\noRemove}| = \noRemove$
\State $\inpSet_{\noRemove} = \lbrace \rbrace$
\For{ $i = 1 \ldots \noRemove$} \label{greedy:iter}
\State $\inpElem =argmin_{\inpElem' \in \inpSet }(\distanceFunction(\mean(\inpSet \setminus \lbrace \inpElem'  \rbrace), \target ) ))$ \label{greedy:dist}
\State $\inpSet_{\noRemove} = \inpSet_{\noRemove}  \bigcup  \lbrace \inpElem \rbrace$
\State $\inpSet = \inpSet \setminus \lbrace \inpElem \rbrace$
\EndFor 
\State return $\inpSet_{\noRemove}$ 
\end{algorithmic}
\caption{\label{algo:greedy} The \greedy\ Algorithm}
\end{algorithm} 
\vspace{-7pt} 
The algorithm works as the follows, at each iteration in Line \ref{greedy:iter}, a point is selected to be removed from the set of remaining points  such that the mean of the remaining points will be close to target \target.
For a collection \inpSet\ of \noPoints\ items, the running time of the \greedy\ algorithm is $O(\noRemove  \noPoints  T_\dimension)$ where \noRemove\ is the number of points to be removed from input data and $T_\dimension$ is the time required to compute the distance \distanceFunction\ in line \ref{greedy:dist} of the algorithm.  In our case, we used euclidean distance and therefore this time is $O(\dimension)$ (number of dimensions). Thus the running  of the \greedy\ algorithm is $O(\noRemove  \noPoints \dimension)$.
Although the \greedy\ algorithms is a simple and fast solution to \LRemoval\ problem; but in practice, it doesn't perform well. Therefore we propose another solution, \cvx\ algorithm, which is presented next.

\spara{The \cvx\ Algorithm:} 
We can formulate \LRemoval\  problem as a Mixed Integer optimization problem to find a binary vector $x$ such that $\mean(\inpSet x) = \target$ 
%(equivalently $\inpSet x= \frac{\target}{n - \noRemove}$)  
subject to $\sum\limits_{i=1}^\noPoints x=\noPoints-\noRemove$. Since solving mixed integer problem is NP-hard and would require complex algorithms to solve, we instead relax this problem to a convex quadratic programming by removing the binary constraints as shown in Algorithm \ref{algo:cvx}. 
This \cvx\ algorithm forms a nonnegative real-valued vector $x$ such that $\distanceFunction(\target, \mean(\inpSet x))$ (Line \ref{algo_cvx:obj}) is small. Note that the aim is to find a subset \subsetS\ of length $\noPoints - \noRemove$ such that $\mean(\subsetS)$ is close to \target, in another word ideally, we want $\noPoints - \noRemove$ elements  of $x$ to be equal $\frac{1}{(\noPoints - \noRemove )}$, and the rest be 0 (and $\inpSet x=\target$).  This constraint is implied in Line~\ref{algo_cvx:ind_q} and~\ref{algo_cvx:sum}. The algorithm tries to find a vector $x$ of real values such that its elements are between 0 and 1  and the sum of elements of x is at least $\noPoints - \noRemove$. Then we transform this real-valued vector to a binary vector by replacing the $ n - \noRemove$ largest elements to 1 and the rest to 0. We used CVX package in Matlab to solve this convex quadratic programming. 

%We can use one of the existing methods to solve this Mixed Integer problem. Since solving mixed integer would require  or use CVX optimization package in Matlab~\cite{grant2008cvx} to solve this problem (with some modifications). The CVX package uses the DCP  (disciplined convex programming) approach which is an effective methodology for organizing and implementing parser-solvers for convex optimization
% DCP has emerged as an effective methodology for organizing and implementing parser-solvers for convex optimization. In DCP, the user combines built-in functions in specific, convexity preserving ways. The constraints and objective must also followcertain rules. As long as the user conforms to these requirements, the parser can easily verify convexity of the problem and automatically transform it to a standard form, for transfer to the solver.
 %We solved \LRemoval\ in MATLAB with the following piece of code in Algorithm \ref{algo:cvx}.
\vspace{-7pt} 
\begin{algorithm}[ht!]
\scriptsize 
\begin{algorithmic}[1]
\Statex {\bf Input:}  Input set \inpSet, target vector \target\ and number of points to be removed \noRemove
\Statex {\bf Output:}  A subset $\subsetS$ such that $|\subsetS| =  n - \noRemove$
%\State  cvx begin quiet
\Indent
%\State     variable $x(n)$ nonnegative
\State      minimize \distanceFunction($\inpSet  * x * \frac{1}{n - \noRemove},  t'$) \label{algo_cvx:obj}
\State     subject to
\Indent
\State     sum($x$) $\geq n - \noRemove $ \label{algo_cvx:sum}
     %for i=1:length(data)
%\State        $x \leq  \frac{1}{(n - \noRemove )}$  \label{algo_cvx:ind_q}
\State    $x \geq 0$ and $x \leq 1$ and $x \in \mathbb{R}$ \label{algo_cvx:ind_q}
\EndIndent
\EndIndent
%\State cvx end
\State Set $ n - \noRemove$ largest values of $x$ to 1 and the rest to 0
\State return $\subsetS = \inpSet (x)$
\end{algorithmic}
\caption{\label{algo:cvx}  The \cvx\  Algorithm}
\end{algorithm} 
\vspace{-7pt} 
\subsection{Solving \targetClustering\ Problem}
Our algorithm for the \targetClustering\ problem, which we call the \maxBenefit\ algorithm, is a polynomial time algorithm that  finds a partitioning of data points \inpSet, into \noCluster\ partitions. The pseudocode of the \maxBenefit\ algorithm is shown in Algorithm \ref{algo:maxbenefit}.
\vspace{-7pt} 
\begin{algorithm}[ht!]
\scriptsize 
\begin{algorithmic}[1]
\Statex {\bf Input:}  Input set \inpSet, target vectors \targetSet\ and number of partitions \noCluster
\Statex {\bf Output:}  Partitions \cluster
\State $\cluster= \{ \}$ \label{algo:p1_start}
\For{\inpElem\ in \inpSet }
\State $j =argmax_{i=1 \ldots \noCluster}( \distanceFunction(\target_i, \mean(\cluster_i)) -  \distanceFunction(\target_i, \mean(\cluster_i \bigcup {\inpElem}) ))$ \label{algo:p1_asssign}
\State Add \inpElem\ to partition $j$, $\cluster_j$
%\State $\clusterCenter_j = mean(\cluster_j)$
\EndFor  \label{algo:p1_end}
\While{no convergence achieved} \label{algo:p2_start}
\For{\inpElem\ in \inpSet }
\State $h = \text{ The partition } \inpElem \text{ belongs to}$ 
\State $loss =   \distanceFunction(\target_h, mean(\cluster_h \setminus  \lbrace\inpElem\rbrace )) -  \distanceFunction(\target_h, mean(\cluster_h ))$
\State   $j = argmax_{i=1 \ldots \noCluster}( loss +  \distanceFunction(\target_i, \mean(\cluster_i ))  - \distanceFunction(\target_i, \mean(\cluster_i \bigcup \lbrace\inpElem\rbrace )) )$  \label{algo:p2_reassign}
\State Remove \inpElem\ from partition $h$ and update $\mean (\cluster_h)$
\State Add \inpElem\ to partition $j$ and update $\mean (\cluster_j)$
%\State $Cost = \sum\limits_{i=1}^\noCluster \distanceFunction( \mean(\cluster_{i}), \target_{i} )$
%\State $\clusterCenter_j = mean(\cluster_j)$
\EndFor
\EndWhile \label{algo:p2_end}
\State Return \cluster
\end{algorithmic}
\caption{\label{algo:maxbenefit} The \maxBenefit\  Algorithm }
\end{algorithm}  
\vspace{-7pt} 
The algorithm works as follows. First, the input data is  partitioned into \noCluster\ partitions in Line~\ref{algo:p1_start}  to Line~\ref{algo:p1_end}. The decision to assign a point to partition $j$ is made on  Line~\ref{algo:p1_asssign} based on how much that point makes the mean of the partition closer to the target point of that partition.  Then after initial partitions are formed,  points are reassigned in   Line~\ref{algo:p2_start}   to  Line~\ref{algo:p2_end} . For each point, the benefit of assigning it to other partitions is computed. This benefit is defined as the gain of assigning that point to other partitions + the loss of removing it from its own partition (Line~\ref{algo:p2_reassign}).  The process is repeated till the solution stabilizes and there is no more improvement in the cost. In practice, the solution stabilized in a few iterations.  
%We actually keep the mean of each partition and whenever a point is added/removed, we update the corresponding partition center, so we do not to recompute partition centers in Line~\ref{algo:p1_asssign} and \ref{algo:p2_reassign}. 
The running of time of each iteration of Algorithm~\ref{algo:maxbenefit} is  $O(\noPoints \noCluster d )$  in which $d$ is the number of dimensions (for input data \inpSet, and targets \target).
  
\textbf{Discussion:} 
It may seem that instead of using Algorithm \ref{algo:maxbenefit}, we can assign each point to its closest target. To see why this algorithm doesn't work, consider three data points $a=(1, 0)$, $b=(-1, 0)$, $c=(-1, 20)$, and two target vectors $t_1=(0, 0)$ and $t_2=(-1, 10)$. If each point is assigned to the closest target, the first partition includes $a$ and $b$ and the second partition only includes $c$. In this case, the centroids of the partitions are $(0, 0)$ and $(-1, 20)$ respectively. On the other hand, our proposed algorithm assigns $a$ to the first partition and $b$ and $c$ to the second partition with centroids as $(1, 0)$ and $(-1, 10)$ which clearly has a lower cost.

\subsection{Solving the \targetLClustering\ Problem} \label{sec:final_algo}
Let's assume we are given \noCluster\ partitions (such that mean of each partition $\cluster_i$ is close to $\target_i$) and along with each partition, we are given $q$ points (for all $q=1\ldots \noRemove$) to be removed from that partition. Now we can develop a dynamic programming algorithm to optimally identify \noRemove\ points to be removed from all partitions.  Let $B(i, q)$ denote the benefit of removing the given $q$ points from the $i^{th}$ partition. The benefit is defined as how much the mean of a partition will get closer to its designated target. More formally the benefit of removing points $\inpSet_q$ from partition $\cluster_i$ with target $\target_i$ is $\distanceFunction(mean(\cluster_i), t_i) - \distanceFunction(mean(\cluster_i \setminus \inpSet_q),t_i)$. Let also $MBR(i - 1, j - q)$ denote the benefit of removing $j - q$ points from partitions $1^{st}$ to $(i-1)^{th}$ partitions. The final goal is to find the $MBR(\noCluster, \noRemove)$, the benefit of optimally removing $\noRemove$ points from first to last ($\noCluster^{th})$ partition. The following dynamic programming shows how we decide upon removing points from partitions optimally.
$$MBR(i, j) = \max\limits_{0 \leq q \leq j} \lbrace MBR(i - 1, j - q) + B(i, q) \rbrace$$
At this stage, we have the tool to remove \noRemove\ points from all partitions. Now we can use Algorithm \ref{algo:maxbenefit} to find the partitions \cluster\ and using Algorithm \ref{algo:cvx}, we can find the $q$ points to be removed from each partition. Putting everything together, we end up with Algorithm \ref{algo:targetLClustering} which is an efficient solution to the \targetLClustering\ problem.  We call this algorithm \finalAlgo.
%Our algorithm for solving \targetLClustering\ problem consists of three steps.  In the first step, we will partition the input into \noCluster\ clusters, using the algorithm for solving \targetClustering\ algorithm (algorithm~\ref{algo:maxbenefit}). Then at step 2 for each partition, we find $i$ points for $i= 1\ldots\noRemove$  to remove from each partition using the algorithm for \LRemoval\ problem (algorithm~\ref{algo:cvx}). Finally, at step 3 we use the following dynamic programming to decide which $\noRemove$ points should be removed from all partitions. 
%$$MBR(i, j) = \max\limits_{0 \leq q \leq j} \lbrace MBR(i - 1, j - q) + B(i, q) \rbrace$$
%$B(i, q)$ is the benefit of removing $q$ points from a cluster $i$ in which benefit is defined as how much the mean will get closer to the target. More formally the benefit of removing points $\inpSet_q$ from cluster $\cluster_i$ with target $\target_i$ is $\distanceFunction(mean(\cluster_i), t_i) - \distanceFunction(mean(\cluster_i \setminus \inpSet_q),t_i)$. 
\vspace{-7pt} 
\begin{algorithm}[ht!]
\scriptsize 
\begin{algorithmic}[1]
\Statex {\bf Input:}  Input set \inpSet, target vectors \targetSet, number of partitions \noCluster\ and number of points to be removed \noRemove
\Statex {\bf Output:}  Partitions \cluster
\State \cluster = Partition  \inpSet\ into \noCluster\ parts using Algorithm~\ref{algo:maxbenefit}  \label{algo_MBR:clustering}
\For{$\cluster_i$ in \cluster } \label{algo_MBR:removing_start}
\For{$j=1 \ldots \noRemove$}
\State $B(i,j)$ = Benefit of removing $j$ points from partition $p_i$ using Algorithm~\ref{algo:cvx}    
\EndFor \label{algo_MBR:remove}
\EndFor   \label{algo_MBR:removing_end}
\For{$i=1 \ldots \noCluster$}
\For{$j=1 \ldots \noRemove$}
\State $MBR(i, j) = \max\limits_{0 \leq q \leq j} \lbrace MBR(i - 1, j - q) + B(i, q) \rbrace$
\EndFor
\EndFor
\State return $MBR(\noCluster, \noRemove)$ and remove corresponding $\noRemove$ from partitions \label{algo_MBR:dp}
\end{algorithmic}
\caption{\label{algo:targetLClustering} The \finalAlgo\  Algorithm }
\end{algorithm} 
\vspace{-7pt} 
The algorithm works as follows:  First, the data is partitioned into \noCluster\ partitions such that the mean of each partition is close to its target point in Line~\ref{algo_MBR:clustering}. Then for each partition,  the $j$ ($j=1\ldots\noRemove$) points that can be removed from that partition is obtained in Line~\ref{algo_MBR:removing_start} to Line~\ref{algo_MBR:removing_end}, and the benefit of removing each $j$ point from each partition is kept  in the matrix $B$. Finally, with a dynamic programming in Line~\ref{algo_MBR:dp} it is decided which $\noRemove$ points should be removed from all partitions. Note that in the algorithm instead of using Algorithm \ref{algo:cvx}  at Line \ref{algo_MBR:remove}, we can also use the \greedy\ algorithm, although the \greedy\ algorithm runs faster, but Algorithm \ref{algo:cvx} outperforms the \greedy\ algorithm. 

Running time of step 1 (Line~\ref{algo_MBR:clustering}) is $O(\noPoints \noCluster d)$. Step 2 (Line~\ref{algo_MBR:removing_start} to Line~\ref{algo_MBR:removing_end}) involves solving $\noCluster \noRemove$ quadratic programming. In theory, the convex quadratic programming can be solved  in the cubic number of variables~\cite{Ye1989}. Thus the  overall running time  of step 2 is $O(\noCluster \noRemove \noPoints^3)$. We have used  the Infeasible path-following algorithm to solve convex quadratic programming which in practice runs much faster than $O(\noPoints^3)$.  Finally, the dynamic programming will take $O(\noCluster \noRemove^2)$. Therefore the overall running time of Algorithm~\ref{algo:targetLClustering} is $O(\noCluster \noRemove \noPoints^3)$. 
In practice, our algorithm tends to have partitions of almost the same size and hence number of data points in each partition is almost $\frac{\noPoints}{\noCluster}$. This means our algorithm runs in $O(\frac{\noRemove \noPoints^3}{\noCluster^2})$ in practice. 
Note that Step 2 in  Algorithm \ref{algo:targetLClustering} can be executed in parallel for all partitions and all values of $j$.

%It should be noted that Algorithm \ref{algo:targetLClustering} is not a new clustering algorithm. \finalAlgo\ algorithm is an attempt to solve Problem  \ref{section:final_prob} which arises in many workforce and educational settings. 
%Mainly the running time of our algorithm is bounded by the running time of the solver in line~\ref{algo_MBR:remove}.  

\section{Experiments}\label{sec:experiments}
In this section, we evaluate our algorithmic solution to the {\targetLClustering} problem. 
\textbf{Our datasets and implementation are immediately and freely available for download.}\footnote{https://github.com/TeamPartitioning/TeamPartitioning} 
We begin with a description of the datasets and baseline algorithms,  
followed by a detailed description of each experiment.

%The goal of these experiments is to gain an understanding of how our  algorithm works in terms of performance (objective function) and runtime. Furthermore, we want to %understand
%how the \noRemove\ parameter impacts our algorithm. We validate our proposed algorithm on a mix of synthetic
%and real datasets. The datasets our available in authors homepage. 
 %We also make the real world dataset, semi synthetic dataset and the code to generate datasets available. 
\vspace{-3pt}   
 \begin{figure*}
\centering
  \includegraphics[scale=0.27]{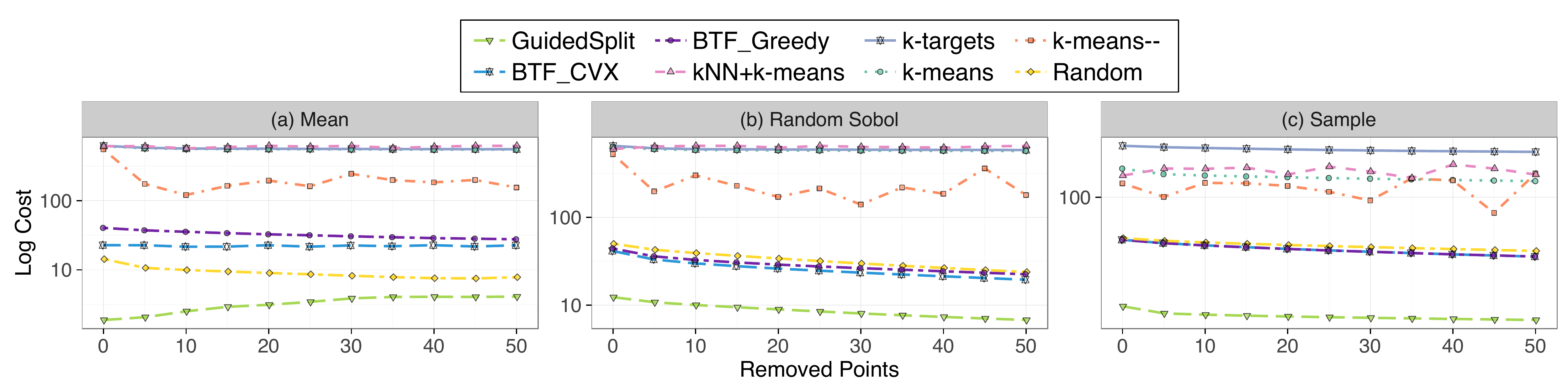} 
  \vspace{-12pt}   
  \caption{Different algorithms on \DSskills\ dataset with respect to increasing \noRemove. The fixed parameters are \noCluster\ = 5, \noPoints\ = 500 and \dimension\ = 10.} \label{fig:expLSkills}
  \vspace{-12pt}  
\end{figure*}
\vspace{-3pt}   
\subsection{Datasets}
We next describe the datasets that we used in our experiments. 
%\subsection{ {\large \synth} Dataset} \label{data:synth}
\begin{itemize}
%e expect our algorithm to be able to find the right clusters of points while other clustering algorithms such as  \kmeans\ cannot find this hidden structure.
%\subsection{{\large \DSskills} Dataset} 
\item The \DSskills\ dataset consists of  %data collected from Linkedin. It includes the profiles of 
the skill sets of 14960 LinkedIn users who have listed at least one data science role in their profiles. %At the time of collection (July 2015), this was a complete dataset of all such individuals on the platform. For each user profile, we collected all the skills in the profile, as well as the number of endorsements that the user had received for each skill. 
We only keep the 40 most frequent skills in the dataset. Each user is thus represented by a vector of length 40, which contains the number of endorsements she has received for each skill from fellow LinkedIn users.
%\subsection{{\large \bia} Dataset}
\item The \bia\ dataset is collected from entry surveys taken by all students who take the Analytics course offered by one of the authors of this paper. %~\footnote{The University and name of of the instructor are removed for anonymity}. 
The data was collected during 6 different semesters and includes 502 students. For each student,
the dataset includes a self-assessment of her level of expertise with respect to machine learning, analytics, programming, and experience with team projects. The assessments
are given on a scale from 0 (no experience) to 3 (very experienced). %, with higher values representing higher expertise.
% \subsection{{\large \guru} Dataset}
\item The \guru\  dataset consists of  6473 experts and 1764 projects from \href{www.guru.com}{www.guru.com}.  
While the majority of projects in Guru require up to 10 skills, larger projects of 30 skills or more are also posted. The dataset also includes the skills required by the projects, which we use to populate
the target vectors. Both the target vectors and the expert skill sets are binary. If an expert possesses a particular skill (or a project requires that skill), the value for that skills is set to 1, otherwise to 0. 
% \subsection{{\large \freelancer} Dataset}
\item The \freelancer\ dataset contains 1763 experts and 721 projects from \href{www.freelancer.com}{www.freelancer.com}. On the Freelancer platform, employers are only allowed to specify at most 5 skills per project. 
Hence, for each project in our dataset, we have a vector in which the value of the five listed skills are set to 1 and the rest to 0. For each expert and each skill they possess, we have a percentage, such that the sum of skill percentages for each expert is equal to 1.
\item The \synth\ dataset is used to generate a ground truth benchmark for our evaluation. We generate  the synthetic data as follows. Let \noCluster\ be the number of partitions, \pointInCluster\ the number of data points per partition, \noRemove\ the number of noisy points added to the data (and should be removed by an effective algorithm), \dimension\ the number of  dimensions, and \std\ the sampling error (standard deviation). In order to generate synthetic data, we first create \noCluster\ target vectors  by sampling from the $[0,1]^\dimension$ value space. Then, using each target vector, we generate \pointInCluster\ data points from a normal distribution with $\mu = \target_i$ and \std. Finally, we add \noRemove\ noisy instances by sampling uniformly from $[0,1]^\dimension$. 
\end{itemize}
\subsection{Methods For Generating the Target Vectors}
For experiments on the \guru\ and \freelancer\ datasets, we use the skills required by the projects to populate the target vectors.
The other datasets, however, do no include a resource that we can intuitively use for the same purpose.
We address this by considering three different methods for generating the target vectors: {\tt Random-Sobol}, {\tt Mean}, and {\tt Sampling}. 
%The {\tt Mean} method simply sets the $i_{th}$ dimension of target vector to be equal to the average value of all the points in the corresponding $i_{th}$ dimension. 
The {\tt Mean} method sets the target vector equal to the average value of all the points in the dataset.
The {\tt Sample} method randomly selects points from the dataset to serve as the target vectors.
The {\tt Random-Sobol} method uses the concept of the Sobol Sequence~\cite{sobol1967distribution} to generate targets with \dimension\ features.
When using a standard uniform random generator, it is possible for all or many of the samples values to  fall within a specific pocket of the sampling space coincidentally. Quasi-random sequences like the Sobol Sequence address this issue, as their output is constrained in order to lead to low-discrepancy samples.  This is achieved by introducing a correlation between the samples (i.e., the generation of a sample considers the values of all the previous samples), which ensures that the sample values are more evenly spread.  
%For the {\tt Random-Sobol} and {\tt Sample} approaches, we generate \rep\ (number of times we repeated each experiment) target vectors and report the average result achieved by each algorithm over all vectors. 
%Note that in this case, all target vectors are equal.  We also consider the case that we create different target vectors for each partition.
%The Sobol Sequence is a quasi-random
%sequence with the following property: for all values of $N$, its subsequence (x_1, ... x_N) has a low discrepancy. The discrepancy of a sequence is low if the %proportion of points in the sequence falling into an interval 
%in which it forms successively finer uniform partitions of the unit interval and then reorder the coordinates in each dimension. A useful property of this sequence is %the points generated by the sequence should fill $[0,1]^\dimension$ while  minimizing the holes. 
\vspace{-3pt}   
  \begin{figure*}
\centering
  \includegraphics[scale=0.27]{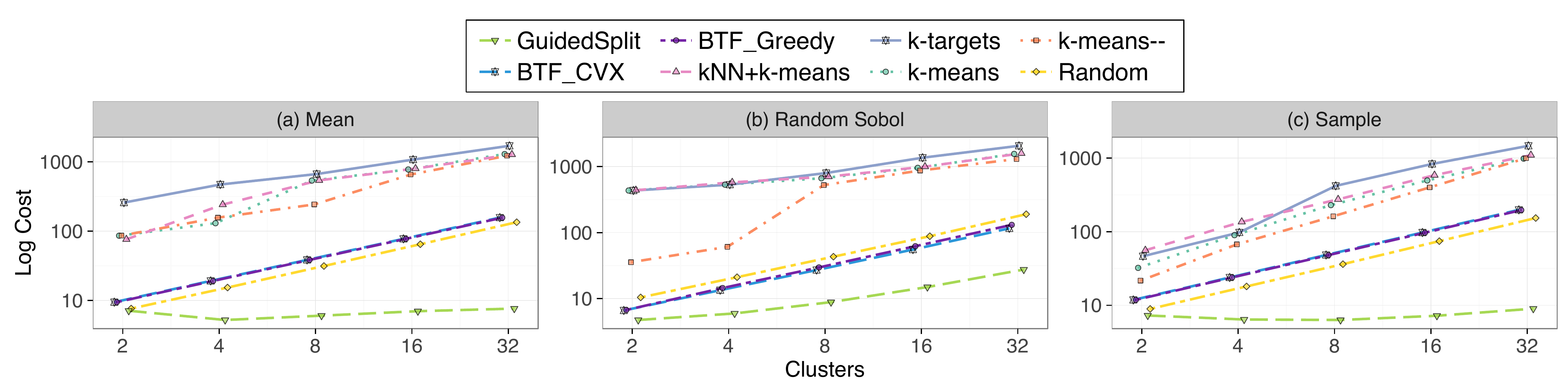} 
    \vspace{-12pt}
  \caption{Different algorithms on  \DSskills\ dataset with respect to increasing \noCluster. The fixed parameters are  \noPoints\ = 500, \noRemove\ = 50 and \dimension\ = 10.} \label{fig:expKSkills}
    \vspace{-12pt}
\end{figure*}
\vspace{-3pt}   
\subsection{Baseline Algorithms}
Next, we describe the baseline algorithms that we compare against our own method.
\begin{itemize}
\item \randomPartitioning: \randomPartitioning\ randomly assigns points to partitions. 
\item \kmeans: \kmeans\ is a clustering method used to minimize the average squared distance between the points in the same cluster. %Solving the \kmeans\ problem~\cite{hartigan1979algorithm} exactly is NP-hard. Lloyd's algorithm~\cite{Lloyd82leastsquares} solves this problem by choosing \noCluster\ centers randomly and assigning the points to the closest center. The centers are then iteratively recomputed based on the points assigned to them. These two phases are repeated until there is no more improvement on the cost of the clustering. After the algorithm converges, we use the Hungarian algorithm to match the obtained cluster centers to the given targets. The same algorithm is also used for the matching phase of all the methods described below.
\item \kmeanst: \kmeanst\ is similar to \kmeans\, with the exception that we used the target vectors to initialize the cluster centers. The motivation is to help the algorithm
converge to a good solution, given its well-documented sensitivity to the initial seed centers~\cite{Arthur:2007}. %Cite: k-means++: The Advantages of Careful Seeding
\item \kmeansmin: \kmeansmin\ ~\cite{ChawlaG13}	 is    a generalization of the \kmeans\ optimization problem, which tries to cluster the data and  discovering the outliers simultaneously. Therefore, the delivered solution consists of \noCluster\ clusters and \noRemove\ outliers. 
\item \knn +\kmeans : In this approach, we computed the nearest neighbors of each point and removed \noRemove\ data points with the largest distance to their nearest neighbors.  After removing these \noRemove\ points, we run \kmeans\  on the remaining data points and get \noCluster\ partitions.
\item \btfCvx\ and \btfGreedy: \textit{Best Team First} is a family of algorithms that work iteratively, and at each iteration creates the best team from the remaining available users. This is a popular technique \cite{anagnostopoulos2012online,Agrawal:2014} that clearly does a excellent job of maximizing the quality of the first group. However; the quality for the subsequent teams can decrease considerably. 
%: There are team formation algorithms that at each time only select one team. We call these family of algorithms Best Team First. In  \cite{anagnostopoulos2012online},  at each time a new team is formed  from all available users  such that the new team posses all the required skill of the current project and the team has small communication overhead.  In \cite{Agrawal:2014}, the goal is to find $k$ teams such that the gain of each team in maximized. The proposed algorithm is an iterative algorithm which at each iteration identifies a team of size $s$ from the remaining data points. Clearly, this aforementioned algorithm does a great job of maximizing the gain for the first group. However; the gain for the subsequent teams can decrease considerably.  
We take a similar approach here, and at each time  we select a team of size $\frac{\noPoints - \noRemove}{\noCluster}$ from the remaining data points. To  select a team of a particular size at each iteration, we use \cvx\ and \greedy\  algorithm where at each iteration, \inpSet\ is the set of remaining data points and the number of points to be removed is the  number of remaining data points - $\frac{\noPoints - \noRemove}{\noCluster}$. We call these baselines  \btfCvx\ and \btfGreedy\ respectively.
%Unfortunately, we are not able to use other team formation algorithms as baselines, as those algorithms optimize very different objective functions and  or are designed for single-team selection, and thus not suitable for our problem. 
\end{itemize}
%We also considered team formation algorithms as baselines (e.g. ~\cite{wang2016truthful}  and~\cite{Agrawal:2014}). However, as these algorithms optimize very different objective functions and  or are designed for single-team selection, they are not suitable for our problem. 
% This algorithm uses the K-nearest neighbors (\knn) approach for distance-based outlier detection to find and remove the \noRemove\ points with the largest distance to that neigbor. After removing these outliers, we run \kmeans\  on the remaining data points. 
%The running of time of the random algorithm is $O(n)$, since it only has to create a random permutation of the points and then create \noCluster\ clusters from members of this permutation. 

After obtaining the partitions using  \randomPartitioning, \kmeans, \kmeanst, \kmeansmin, and \knn +\kmeans\  approaches, we need to find a correspondence between each partition and each target vector. The task of assigning partitions to specific target vectors is an instance of assignment problem. The input to the problem includes two sets of equal sizes, \targetSet\ (given targets) and \clusterCenters\ (obtained partition centers), and 
a weight function $\distanceFunction : \targetSet \times \clusterCenters \rightarrow R$. The goal is to find a bijection $\pi: \clusterCenters  \rightarrow \targetSet$, such that the cost function $\sum_{a\in \clusterCenters} \distanceFunction(a, \pi(a))$ is minimized. We can complete this task via The Hungarian algorithm~\cite{edmonds1972}, which solves the assignment problem optimally in $O(\noCluster^3)$, where $|\clusterCenters| = |\targetSet| = \noCluster$.  Thus, after obtaining the partitions, we use the Hungarian algorithm to match the obtained partition centers to the given targets. 
\vspace{-3pt}   
  \begin{figure}[htb]
\centering
  \includegraphics[scale=0.25]{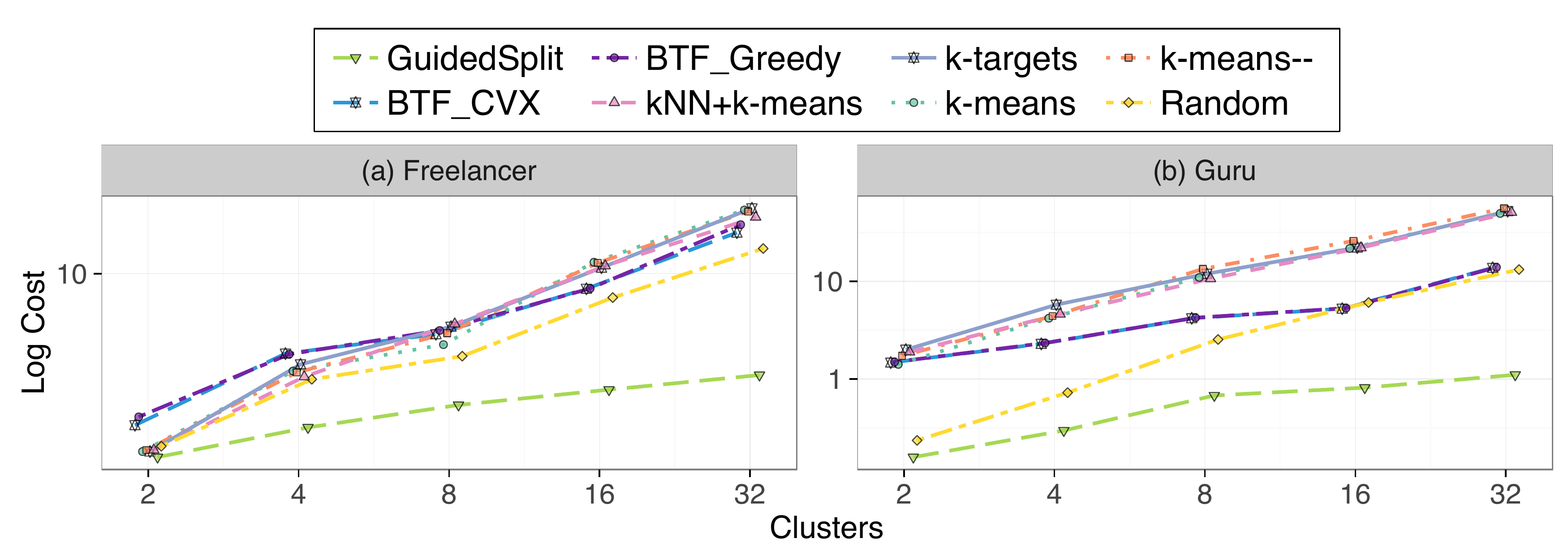} 
    \vspace{-15pt}
 \caption{Different algorithms on \freelancer, and \guru\ datasets. } \label{fig:exp_free_guru}
   \vspace{-15pt}
\end{figure}
%%%% for free guru
\vspace{-3pt}   
\subsection{Experimental Evaluation}
Having all the datasets and target vectors ready as the input of Problem \ref{section:final_prob}, we compare the efficacy of \finalAlgo\ algorithm with the baseline algorithms.  We implemented our algorithm in Matlab and used CVX package to solve the convex optimization problem in the \cvx\ algorithm. All experiments were carried out on a machine with a 2.4 GHZ CPU, 16 GB RAM, running CentOS Linux 7.  We repeated each experiment 25 times and reported the average. The confidence intervals were very tight, and we thus omit them from  the plots. We compare the cost of our algorithm compare to all other baselines. This cost refers to the distance from the mean of each partition to its target vector, see equation \ref{eq:cost}. What follows is the detailed description of each experiment.
%REMOVE THIS:
%The results are shown in figures \ref{fig:expLSkills}, ref{fig:expKSkills},  and  in which we vary parameters \noRemove, \noCluster\ and \dimension. In each figure the %rest of the parameters remain fixed and they are reported on the caption of the figure.

%In Random Sobol  we have used Sobol Sequence~\cite{sobol1967distribution} to generate \rep\ vectors of \dimension, where \rep\ is the number of iterations we have repeated the experiments and reported the average. Sobol sequence is a quasi-random low-discrepancy sequence in which it forms successively finer uniform partitions of the unit interval and then reorder the coordinates in each dimension. A good property of this sequence is the points generated by the sequence should fill $[0,1]^\dimension$ while  minimizing the holes. For each iteration we used one of these vectors and set the targets of all clusters to that vector. In Mean method we set the target of all clusters to the mean of the datatset. Finally in Sample method, we sampled \rep\ points from the dataset and at each iteration we set all targets to one of these samples. 
%
%
%
\subsubsection{Effectiveness with respect to number of discarded points \noRemove}
The goal of this experiment is to demonstrate how different algorithms behave with respect to the number of removed points \noRemove, as presented in the definition of Problem~\ref{section:final_prob}. 
In this experiment we set number of data points to 500, number of partitions to 5, number of dimensions to 10  (\noCluster\ = 5, \noPoints\ = 500 and \dimension\ = 10). 
%Figure~\ref{fig:expLSkills} shows the cost (in log scale) versus \noRemove\ on the \DSskills\ dataset. We also set \noCluster\ = 5, \noPoints\ = 500 and \dimension\ = 10. 
%We have normalized all the values between 0  and 1 and have 
 We report the results of all three different methods of choosing the target vectors ( {\tt Random-Sobol}, {\tt Mean}, and {\tt Sampling} methods).
Figure~\ref{fig:expLSkills} shows the cost (in log scale) versus \noRemove.

We observe that our \finalAlgo\  algorithm consistently outperforms the baseline algorithms for all target vectors, and the cost of  \finalAlgo\ is significantly less than all the baselines. 
We also observe  when all the targets are set to the mean, the \randomPartitioning\ algorithm is doing  well. This is simply because a random sample of the data is expected to have the same mean as the whole data.
%\hl{VERY CONFUSING text: 

An interesting observation is that when targets are equal to the mean of the dataset, Figure~\ref{fig:expLSkills}(a), as \noRemove\ increases the cost of \finalAlgo\ increases as well. The reason is after finding partitions when \noRemove\ = 0, \finalAlgo\ finds near perfect solution  in which mean of each partition is very close to its target. Removing points from the partition discomposes the obtained solution and makes the mean of partition far from its target.
We also tried this experiment  
%with different target vectors for each partition as well as 
on other datasets; the results were similar, and \finalAlgo\ outperforms in those experiments as well. However, the algorithms were more competitive on \DSskills\ dataset. Thus due to lack of space, we only report results on  \DSskills\ dataset.
%We also experiment the cost versus \noRemove\ on the \synth\ dataset. In this experiment we set \noCluster = 10, \pointInCluster = 100 and \dimension = 10. We generated \noCluster\ targets by sampling from space $[0,1]^\dimension$ using Sobol sequence and we used these targets as the desired targets of the clusters. Then using each target we generated \pointInCluster\ points. For each $\noRemove=1\ldots 50$ we added \noRemove\ noisy points to the dataset \inpSet.
% The results on  \synth\ datasets is shown in figure \ref{fig:expLSynth}.  The results illustrate that \finalAlgo is capable of forming the right clusters (with the hidden structures). We also examined the points removed by our algorithm, we observed that majority of removed points are exactly the noisy points we added to the data (the points randomly sampled $[0,1]^\dimension$).  
%We also tried this experiment with different target vectors for each partition as well as on other datasets; the results were similar and due to lack of space the results of those experiments are left out. 
\subsubsection{Effectiveness with respect to number of partitions \noCluster}
The goal of this experiment is to demonstrate how the algorithms behave with respect to the number of partitions \noCluster. In these experiments, we show the cost (in log scale) versus different number of partitions (\noCluster\ = 2, 4, 8, 16, 32) and we fixed \noPoints\ = 500, \noRemove\ = 50 and \dimension\ = 10 on \DSskills\ dataset  and \noPoints\ = 502, \dimension\ =4 and \noRemove\ = 50 for \bia\ dataset. In the experiment on \bia\ dataset, to form fair teams of students with an equal level of expertise, the targets of partitions are set to be equal to the mean of the whole dataset. 
 We present the results in Figure~\ref{fig:expKSkills}(a), Figure~\ref{fig:expKSkills}(b), and Figure~\ref{fig:expKSkills}(c)  for all three different methods of choosing the target vectors ( {\tt Random-Sobol}, {\tt Mean}, and {\tt Sampling} methods) on \DSskills\ dataset  and Figure~\ref{fig:expnbia} depicts the results on \bia\ dataset.  

As observed before,  our \finalAlgo\  algorithm consistently outperforms the baselines for all target vectors, and the cost of  \finalAlgo\ is significantly less than all the baselines. Note that  as the objective function here is to minimize the distance of targets to the mean of partitions, cost increases with the number of partitions. But  \finalAlgo\  is still able to find efficient partitions compared to all baselines, and its cost does not increase dramatically. 
%As the number of clusters increase the cost increases dramatically for all algorithms except \finalAlgo. In a regular \kmeans\ clustering as the number of clusters increases, the \kmeans\ algorithm is able to form a more customized cluster in which the cost of each point to its cluster representative is small. But as the objective function here is to minimize the distance of targets to the mean of clusters, cost increases with the number of clusters. But as \finalAlgo\ is still able to find efficient clusters, its cost does not increase dramatically. 
 %\finalAlgo\ performs really well when targets are set to the mean of the dataset or targets are samples from the data. On the other hand, when targets are chosen by Random Sobol method the targets might be very far from the actual mean of data and makes it harder for \finalAlgo\ to form efficient clusters. 
%and the fixed parameters in the experiment on \synth\ dataset are \pointInCluster=100, \std = 0.1, \dimension =10 and \noRemove = 50. In both experiments, 
%Our result illustrates \finalAlgo\  algorithm outperforms baseline algorithms in \bia\ dataset as well and it minimizes the cost efficiently. 
%The result of algorithms on %\synth\ and 
%\bia\ dataset is shown in figure \ref{fig:expKsynthbia}.(a). % and figure \ref{fig:expKsynthbia}.(b) respectively. 
%We also tried this experiment on \synth, \freelancer, and \guru\ datasets; the result were similar and due to lack of space, we leave out the results on those datasets. 
\vspace{-3pt}   
%\hl{Where is the discussion of the results?}
\begin{figure}[htb]
\centering
\includegraphics[scale=0.25]{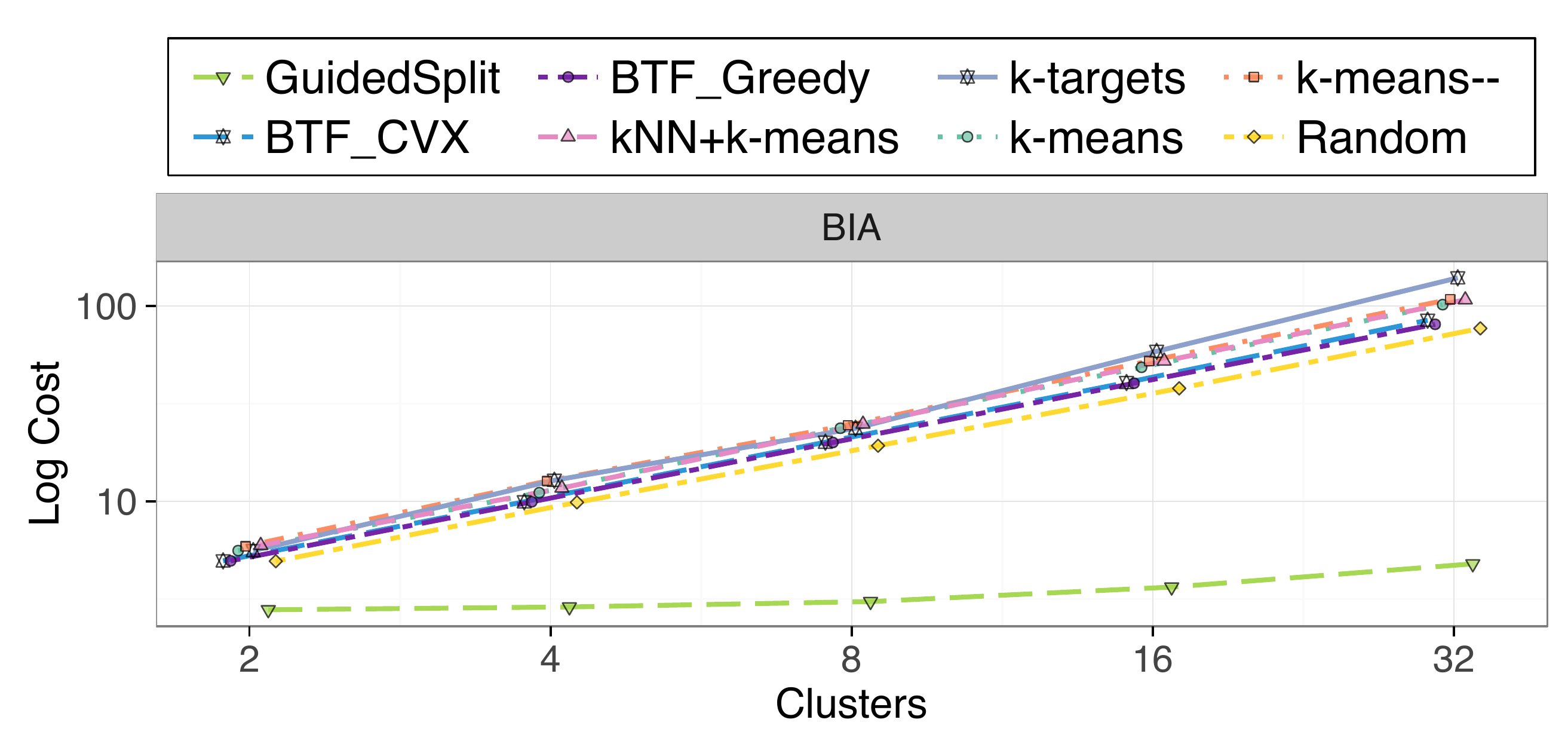}
  \vspace{-15pt}
    \caption{Different algorithms on \bia\  dataset with respect to increasing \noCluster. The fixed parameters are \noPoints\ = 502, \noRemove\ = 50, and \dimension\ = 4.}
      \vspace{-15pt}
    \label{fig:expnbia}
\end{figure}
\vspace{-3pt}   
\subsubsection{Results on the \freelancer\ and \guru\ datasets}
\guru\ and \freelancer\ datasets have a different dynamic compare to other datasets. In these datasets, we can use the required skills of the projects  as the target vectors. Besides, the projects posted in freelancer and guru, are the tasks that require a few skills and they can also be completed by a small team of experts. So we conducted an experiment specifically for these two datasets which indicates  the performance of our algorithm  with respect to the number of points \noPoints, the number of partitions \noCluster, and the number of points to remove \noRemove. 

For this experiment, we used the set of experts as input set  \inpSet, and the projects as target vectors \targetSet. We tried \noCluster\ =2, 4, 8, 16, 32,  and for each value of \noCluster, we used \noCluster\  projects at random.  We  also set \noPoints\ = \noCluster\ * 5 and \noRemove\ = \noCluster\ (and selected \noPoints\ + \noRemove\ experts as random as the input \inpSet). The intuition behind is due to the inherent nature of projects posted in Guru and Freelancer; these projects are small projects that usually do not need more than 5 people to be completed. By setting \noPoints = \noCluster\ * 5 and \noRemove\ = \noCluster, ideally, the size of each team would be 4.
The results in Figure \ref{fig:exp_free_guru}(a) and \ref{fig:exp_free_guru}(b) illustrate \finalAlgo\  algorithm outperforms baseline algorithms on \freelancer\ and \guru\ datasets, and find teams of experts whose proficiencies are close to required skills of projects. 
%\hl{where is the discussion of the results?}
\subsubsection{Effectiveness with respect to population size \noPoints}
We also experiment how different algorithms behave with respect to the number of points \noPoints, for \noPoints\ =100, 500, 1000, 2500, 5000, 10000. In this experiment, we show the cost (in log scale) versus population size, \noPoints on the \synth\ dataset, where the fixed parameters are \noCluster\ = 5, \noRemove\ = 50, and \dimension\ = 10. The result is depicted in figure~\ref{fig:expnSynth}.  

The results illustrate that \finalAlgo\ is capable of forming the right clusters with the hidden structures as was imposed by the \synth\ dataset.  As explained before, 
the \synth\ dataset has a distinct structure; each cluster has a different mean, and this structure  is very easy for \kmeans\, and \kmeanst\ algorithms to find. This is also the reason \randomPartitioning\  performance degrades; because it cannot find the hidden structure of the data.  

In addition, We also examined the points removed by our algorithm; we observed that majority of removed points are exactly the noisy points we added to the data (the points randomly sampled from $[0,1]^\dimension$ space).
We also tried this experiment on other datasets; the results were similar.  However, the algorithms were more competitive on \synth\ dataset compare to other datasets. Thus we only report results on the \synth\ dataset, and we leave out the results on other datasets.  
\vspace{-3pt}   
\begin{figure}[htb]
\centering
\includegraphics[scale=0.32]{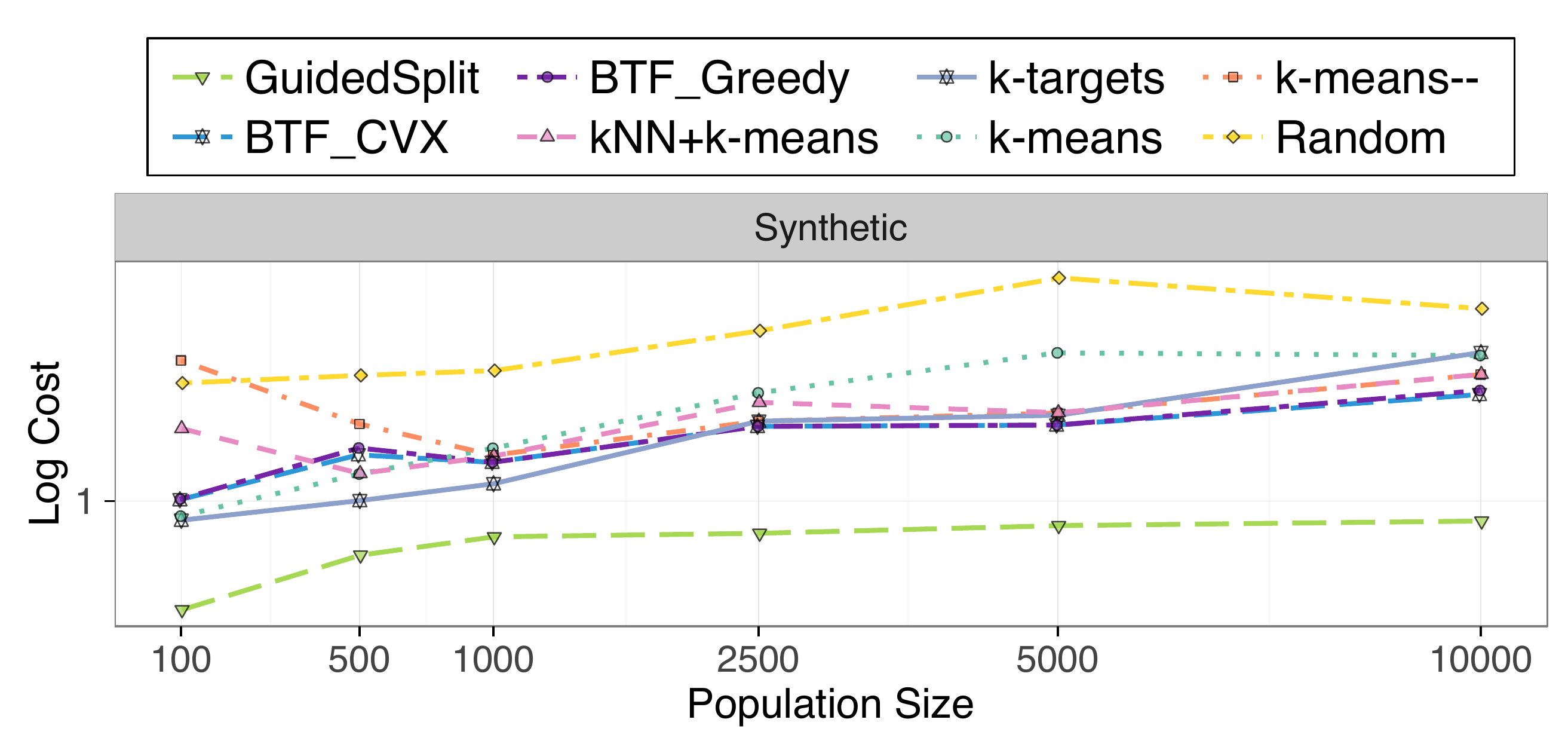}
    \caption{Different algorithms on \synth\ dataset with respect to increasing \noPoints. The fixed parameters are \noCluster\ = 5, \noRemove\ = 50, and \dimension\ = 10.}
    \label{fig:expnSynth}
\end{figure}
\vspace{-3pt}   
\vspace{-3pt}   
 \begin{figure}[htb]
\centering
  \includegraphics[scale=0.25]{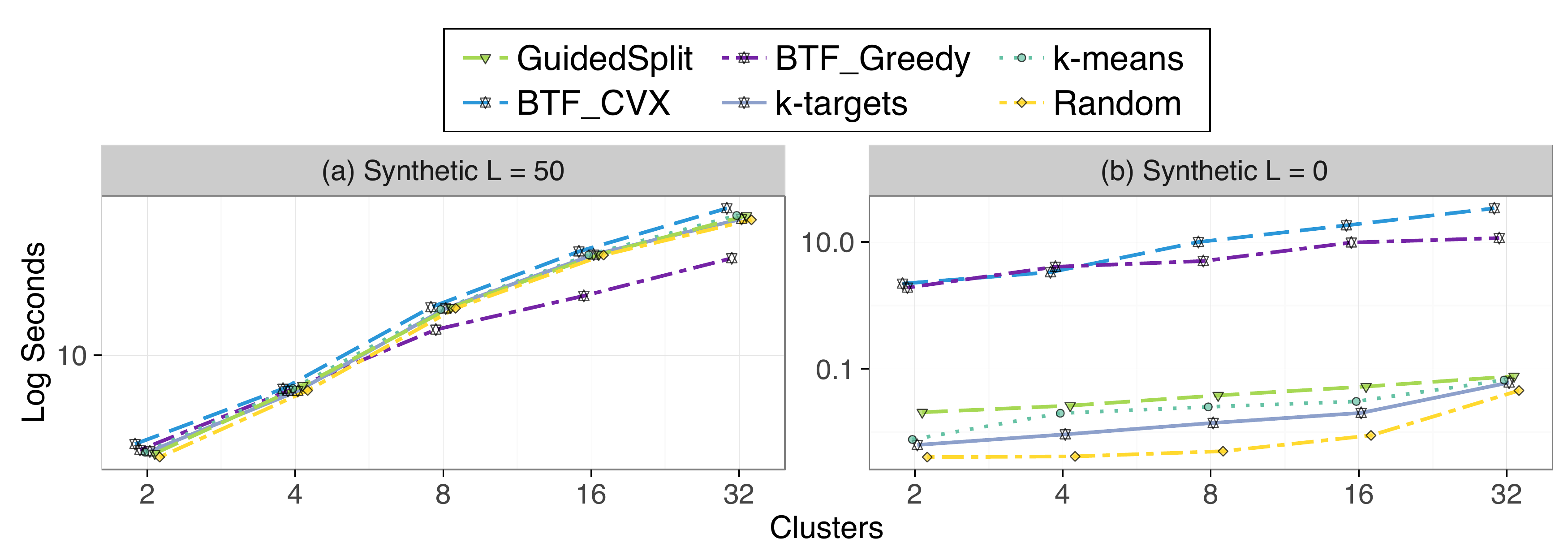} 
    \vspace{-15pt}
 \caption{Running time of algorithms with respect to number of clusters.}\label{fig:runtime_k}
   \vspace{-15pt}
\end{figure}
\vspace{-3pt}   
\subsubsection{Running Times}
%The running time of our algorithm is shown in Figure \ref{fig:runtime_n} and \ref{fig:runtime_k}; 
We only show the running time for the top six  competitive algorithms which performed well in experiments on real datasets and used the \cvx\ or the \greedy\ algorithm. We exclude \kmeansmin\ and \knn +\kmeans\ in this experiments since they do not need to use the \cvx\ or the \greedy\ algorithm  and thus their  running times are less than the other algorithms. Clearly  as illustrated in preceding experiments, the performance of \kmeansmin , and \knn +\kmeans\ are not comparable to other algorithms either. 

 We report the running time on the \synth\ dataset when increasing the number of partitions \noCluster\  where other  parameters are \noPoints\ = 500, \std\ = 0.2 \noRemove\ = 50. We also report the same experiment when \noRemove\ = 0.  Naturally when \noRemove\ = 0, algorithms run faster since they do not need to execute \cvx\ algorithm. This means the running time of \finalAlgo\ algorithm only includes of running Line \ref{algo_MBR:clustering} in Algorithm \ref{algo:targetLClustering}. We also  report the running time versus population size \noPoints\ where fixed parameters are \noCluster\ = 5,  \std\ = 0.2  and for both \noRemove\ = 0 and  \noRemove\ = 50.

  Figure \ref{fig:runtime_k}(a) shows the runtime in log seconds versus increasing the number of clusters where fixed parameters are   \noPoints\= 500 and \noRemove\ = 50, and \std\ = 0.2 . As mentioned in subsection \ref{sec:final_algo}, the running time of our \finalAlgo\ algorithm is mainly bounded by Line \ref{algo_MBR:removing_start} to Line \ref{algo_MBR:remove} in Algorithm \ref{algo:targetLClustering}, where $j$ (for $j = 1 \ldots \noRemove$) points are removed from each cluster. This also applies to \randomPartitioning, \kmeans, and \kmeanst\ since they also use the same algorithm to remove \noRemove\ points.  The running time of Line \ref{algo_MBR:removing_start} to Line \ref{algo_MBR:remove} is $time(\text{removing } \noRemove^2) * \noCluster$. In theory $time(\text{removing } \noRemove)$ is $O(n^3)$ which $n$ is the number of points in each cluster.
In consequence, as the number of clusters increases,  \targetLClustering\ needs to remove $\noRemove ^2$ points for more clusters and hence the running time also increases.  

Note that \randomPartitioning, \kmeans, and \kmeanst\ require an extra step in which after clustering is done, the center of each cluster is matched with one of the targets using Hungarian algorithm. The running of this step only depends on the number of clusters. Figure \ref{fig:runtime_k}(b) shows the running time of the algorithms when \noRemove\ = 0; hence there is no need to run \cvx\ algorithm (except for \btfCvx\ baseline) and the reported running time only includes clustering and using Hungarian algorithm to match  the cluster centers to targets. As an interesting result, the  running time of \btfCvx\ and \btfGreedy\ is considerably higher than other algorithms even though \noRemove\ = 0 and this is due the fact that for the \btfCvx\ (\btfGreedy)  algorithm we need to run the \cvx\ (\greedy) algorithm \noPoints\  times. 

Figure \ref{fig:runtime_n}(a) shows the runtime in log seconds versus increasing the number of points in the data, where \noCluster\ = 5 and \noRemove\ = 50, and \std\ = 0.2. The reason \randomPartitioning\ has the minimum running time is that it creates very balanced clustering in which the number of points in all clusters are very similar. This means when removing the points, the running time of Line \ref{algo_MBR:removing_start} to Line \ref{algo_MBR:remove} in Algorithm \ref{algo:targetLClustering} is  almost $\noCluster * (\frac{\noPoints}{\noCluster})^3$ which is considerably smaller in the case that there is one very large cluster and the rest of clusters have only a few points. On the contrary,  the opposite happens in \kmeans ; \kmeans\ creates a very large cluster including most of the points and assigns the added noises to other clusters; this is the reason \kmeans\ has the largest running time. 

On the other hand, when the means of clusters are given as the initial seed to \kmeanst , it creates more structured and hence more balanced clusters. As a result, \kmeanst\  has a smaller running time compared to \kmeans . \kmeans\ and \kmeanst\ have strange behavior for \noPoints\ = 5000. We repeated experiments many times and got the same results,  we observed for \noPoints\ = 5000, these two algorithms create highly unbalanced clusters, and this might be the reason that the running time is very high for \noPoints\ = 5000. Figure \ref{fig:runtime_k}(b) shows the running time of the algorithms when \noRemove\ = 0. As the results indicate, while our proposed algorithm is very efficient in optimizing the cost, it is also a fast algorithm and its running time is comparable to other alternatives.
\vspace{-3pt}   
 \begin{figure}[htb]
\centering
  \includegraphics[scale=0.25]{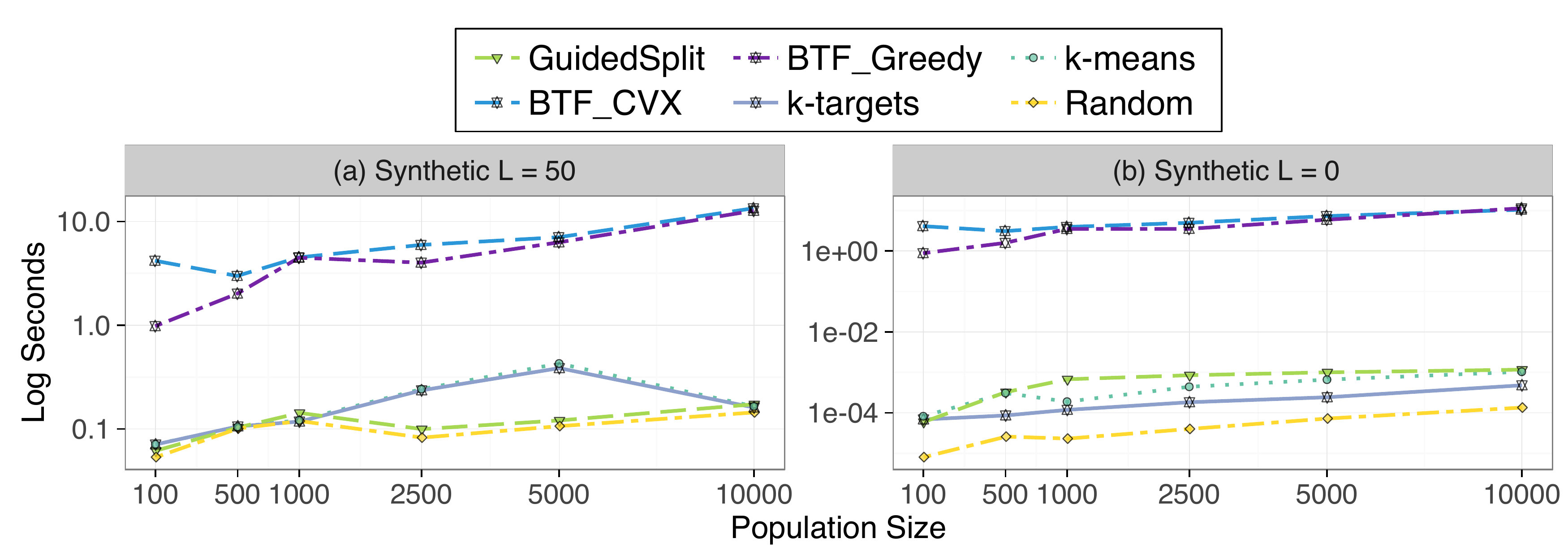} 
    \vspace{-15pt}
 \caption{Running time of algorithms with respect to Population size.}
   \vspace{-15pt}
    \label{fig:runtime_n}
\end{figure}
\vspace{-3pt}   
%\begin{figure}[htb]
%\centering
%\includegraphics[scale=0.25]{./plt/runtime_both_part1.pdf}
  %  \caption{Running time of algorithms with respect to number of clusters and Population size when \noRemove=0}
    %\label{fig:runtimePart1}
%\end{figure}

%\section{Discussion}\label{sec:discussion}
%\input{discussion}
\vspace{1ex}
\section{Conclusions}\label{sec:conclusion}
In this paper, we introduced {\targetLClustering}, a new problem that repeatedly emerges in educational and organizational settings and cannot be addressed by existing team formation algorithms. From the computational viewpoint, we studied the computational complexity of the \targetLClustering \ problem and we proved that the problem is NP-hard to solve and approximate. We also proposed novel and efficient heuristic algorithms, which we evaluated via extensive experiments on both synthetic and real-world datasets. 
The results indicate our methodology is consistently able to deliver high-quality solutions and outperform multiple competitive baselines. Our work provides a platform for future research on problems

%ACKNOWLEDGMENTS are optional

\bibliographystyle{abbrv}
% sigproc.bib is the name of the Bibliography in this case

% ACM needs 'a single self-contained file'!
%\balancecolumns

\balancecolumns
% That's all folks!
\end{document}